\documentclass[11pt]{article}
\usepackage[english]{babel}
\usepackage[utf8x]{inputenc}
\usepackage{libertine}
\usepackage{amsmath}
\usepackage{amsthm, amssymb}
\usepackage{amsfonts}
\usepackage{enumerate}
\usepackage{algorithm}
\usepackage{algorithmic}
\usepackage{graphicx}
\usepackage{xspace}
\usepackage[margin=1in]{geometry}

\usepackage[
  pagebackref,
  colorlinks=false,
  urlcolor=blue,
  linkcolor=blue,
  citecolor=blue,
]{hyperref}
\usepackage[nameinlink]{cleveref}

\allowdisplaybreaks

\newtheorem{theorem}{Theorem}
\newtheorem{conjecture}[theorem]{Conjecture}
\newtheorem{lemma}[theorem]{Lemma}
\newtheorem{prop}[theorem]{Proposition}
\newtheorem{corollary}[theorem]{Corollary}

\newtheorem{mainresult}{Main Result}

\newtheorem{observation}[theorem]{Observation}

\newtheorem{proposition}[theorem]{Proposition}

\theoremstyle{definition}
 \newtheorem{definition}[theorem]{Definition}

\newcommand{\e}{\mathrm{e}}
\newcommand{\E}{\mathbb{E}}

\newcommand{\IGNORE}[1]{}

\newcommand{\A}{\ensuremath{\mathcal{A}}\xspace}

\DeclareMathOperator*{\argmin}{arg\,min}

\newcommand{\OPT}{\mathsf{OPT}}
\newcommand{\ALG}{\mathsf{ALG}}

\title{Formal Barriers to Simple Algorithms for the Matroid Secretary Problem}
\author{Maryam Bahrani\thanks{Columbia University. Email: {m.bahrani@columbia.edu}.} \and%
Hedyeh Beyhaghi\thanks{Toyota Technological Institute at Chicago. Email: {hedyeh@ttic.edu}.} \and%
Sahil Singla\thanks{Georgia Institute of Technology. Part of the work done while at Princeton University. Email: {ssingla@gatech.edu}.} \and%
S. Matthew Weinberg\thanks{Princeton University. Email: {smweinberg@princeton.edu}.}}
\date{}
\begin{document}

\maketitle

\begin{abstract}
Babaioff et al.~\cite{babaioff2007matroids} introduced the matroid secretary problem in 2007, a natural extension of the classic single-choice secretary problem to matroids, and conjectured that a constant-competitive online algorithm exists. The conjecture still remains open despite substantial partial progress, including constant-competitive algorithms for numerous special cases of matroids, and an $O(\log \log \text{rank})$-competitive algorithm in the general case.

Many of these algorithms follow principled frameworks. The limits of these frameworks are previously unstudied, and prior work establishes only that a handful of particular algorithms cannot resolve the matroid secretary conjecture. We initiate the study of impossibility results for frameworks to resolve this conjecture. We establish impossibility results for a natural class of greedy algorithms and for randomized partition algorithms, both of which contain known algorithms that resolve special cases.
\end{abstract}

% !TeX root = main.tex
% !TEX root = main.tex

\section{Introduction} 
The problem of finding a max-weight basis of a matroid $\mathcal{M}=(V,\mathcal{I})$\footnote{Given a finite set $V$ and a family of subsets of $V$ called $I$, we say $\mathcal{M}=(V,\mathcal{I})$ is a matroid if it satisfies (i) $\emptyset\in\mathcal{I}$, 
(ii) \emph{Hereditary Property (downwards closed):} $\forall T\subseteq S\subseteq V$, set $S\in \mathcal{I}$ implies $T\in\mathcal{I}$, and
 (iii) \emph{Exchange Property:} For any $S,T\in\mathcal{I}$ where $|S|>|T|$,  there exists some $x\in S$ such that $T\cup\{x\}\in\mathcal{I}$.
 } is central in the field of combinatorial optimization  (see books~\cite{Oxley-Book06,Schrijver-Book03,KV-Book08}). More specifically, each element $e \in V$ has a weight $w(e) \geq 0$, and the goal is to find the set $S \in \mathcal{I}$ maximizing $w(S):=\sum_{e \in S} w(e)$. Seminal works of Rado, Gale, and Edmonds establish that the following simple greedy algorithm finds a max-weight basis of a matroid $(V,\mathcal{I})$: Initialize $A = \emptyset$, then process the elements of $V$ in decreasing order of $w(e)$, adding to $A$ any element such that $A \cup \{e\} \in \mathcal{I}$~\cite{rado1957note, gale1968optimal, edmonds1971matroids}. In fact, if for some $(V,\mathcal{I})$ this algorithm is optimal for all $w(\cdot)$, then $(V,\mathcal{I})$ must be a matroid.

While simple, this algorithm still requires knowledge of all weights up front. Motivated by applications to mechanism design and other online problems~\cite{HajiaghayiKP04,BIKK-Sigecom08},  recent work considered the problem in an online setting: elements are still processed one at a time and are immediately and irrevocably accepted or rejected upon processing, but an element's weight remains unknown until the element is processed. In particular, the algorithm does not have control over the order of elements and therefore cannot run the simple greedy algorithm.

For a fully adversarial order, it's folklore that the best algorithm can do no better than simply selecting a random element. Babaioff et al.~\cite{babaioff2007matroids}\footnote{Conference version~\cite{babaioff2007matroidsSODA} appeared in 2007.} therefore introduced the \emph{Matroid Secretary Problem} (MSP), where elements arrive in a \emph{uniformly random order} (while the weight function is still adversarial). This formulation extends the classic single-item secretary problem~\cite{dynkin1963optimum}.

Consider an algorithm \A for the matroid secretary problem on matroid $\mathcal{M}$. Let $\OPT$ be the max-weight basis of $\mathcal{M}$ under $w(\cdot)$, and let $\ALG$ be the set of elements chosen by \A (under $w(\cdot))$. The following notion of \emph{utility-competitiveness} for a  matroid secretary algorithm was  studied in Babaioff~{et al.}~\cite{babaioff2007matroids}.

\begin{definition}[Utility-Competitive]
An algorithm \A is \emph{$\alpha$-utility-competitive} if $\E[w(\ALG)]/w(\OPT) \geq \alpha$, where the expectation is over the randomness of the arrivals and any internal randomness of algorithm \A.
\end{definition}

In the same paper that introduced the matroid secretary problem, Babaioff~{et al.}~\cite{babaioff2007matroids} conjecture that there is a constant-utility-competitive algorithm. The stronger form of the conjecture is that this constant is $1/e$.

\begin{conjecture}[Matroid Secretary]\label{conj:babaioff}
There is an $\Omega(1)$-utility-competitive algorithm for the matroid secretary problem.
\end{conjecture}

Despite extensive follow-up work, this conjecture still remains open. Many constant-utility-competitive algorithms have been proposed for specific classes of matroid (see related work in Section~\ref{sec:related}). For general matroids, however, the best known algorithms are $1/O(\log\!\log r)$-competitive~\cite{lachish2014log,FSZ-SODA15} (here, $r$ denotes the \emph{rank} of the matroid, which is the size of the largest set in $\mathcal{I}$).

As the only known lower bound, even for general matroids, is the same $1/e$ from the classic single-item setting, and because Dynkin's algorithm guarantees a stronger property that the heaviest element is selected with probability $1/e$, the following stronger notion of \emph{probability-competitive} algorithms has been also studied~\cite{HoeferK-ICALP17,SotoTV-SODA18}.

\begin{definition}[Probability-Competitive]
An algorithm \A  is \emph{$\alpha$-probability-competitive} if for all $i\in \OPT$ it satisfies that $\mathbb{P}[i\in \ALG] \geq \alpha$.
\end{definition}

\noindent Note that probability-competitiveness is a stronger notion than utility-competitiveness, since the former implies the latter with the same competitive ratio. Soto et al.~\cite{SotoTV-SODA18} showed that many (but not all) existing utility-competitive algorithms can be extended to obtain probability-competitive algorithms. This results in the following more ambitious conjecture. Again, the stronger version conjectures that this constant is $1/e$.

\begin{conjecture}\label{conj:soto} 
There is an $\Omega(1)$-probability-competitive algorithm for the matroid secretary problem.
\end{conjecture}

Progress on both conjectures has been slow. Indeed, even the strong version of Conjecture~\ref{conj:soto} remains plausible, while the best utility-competitive algorithms have stalled at $1/O(\log\log r)$~\cite{lachish2014log,FSZ-SODA15}. One thesis motivating our work is that the community currently lacks structure for narrowing a search among numerous promising approaches. Existing algorithms for special cases indeed follow principled frameworks, but these frameworks are quite flexible and it remains unknown which (if any) of them might produce a resolution to either conjecture. 

One particularly enticing possibility is that a simple ``greedy-like'' algorithm might even work. Note that such algorithms indeed work in the Free-Order model~\cite{JailletSZ13}, or for the related Matroid Prophet Inequality~\cite{kleinberg2012matroid}, or for special cases of the Matroid Secretary Problem~\cite{dynkin1963optimum,babaioff2007knapsack}. There are numerous variants of ``greedy'' algorithms, though. While many particular variants are known to fail on the same ``hat graph''~\cite{babaioff2007matroids}, there is previously no approach to quickly tell whether a novel greedy variant is already known to fail.

In this work, we rigorously consider two general classes of algorithms, and prove super-constant lower bounds on what they can achieve for the matroid secretary problem. This both helps explain why these types of algorithms have faced difficulty extending beyond the special cases for which they were originally designed, and helps guide future work towards precisely the variants that merit further exploration.

%---------------------------------------------------------------------------
\subsection{Greedy Algorithms}

Since finding the max-weight basis of matroids without requiring irrevocable commitments can be done exactly by the simple greedy algorithm, the class of greedy algorithms is a very natural candidate for solving the Matroid Secretary Problem.  We consider a large family of ``greedy-like'' algorithms. We define three natural properties that a greedy algorithm might have, and establish that any algorithm satisfying these properties cannot be constant-utility-competitive (Theorem~\ref{thm:greedy}). We postpone formal statements of the properties until Section~\ref{ch:prophet}, but overview them here: (i) the algorithm should reject the first $T$ fraction of elements, (ii) the algorithm at all times stores an independent set $I$ containing all accepted elements and no elements rejected after $T$, (iii) an element is accepted if and only if it improves the max-weight basis of $I$ after contracting the accepted elements.\footnote{To rephrase (iii), an element $e$ is accepted iff after contracting the accepted elements (not including $e$), the max-weight basis of the restricted matroid to $I\cup\{e\}$ is heavier than the max-weight basis of the restricted matroid to $I$ (the latter being exactly the weight of $I$ since $I$ is independent).}
Note that this a general framework rather than a fully-specified algorithm, since it allows for the algorithm to choose $I$ (it need not be the max-weight basis after contracting the accepted elements, just some independent set). 

In Section~\ref{ch:prophet} we overview several existing algorithms that fit this framework, and Theorem~\ref{thm:greedy} unifies a proof that none of these algorithms (or many hypothetical ones) can be constant-utility-competitive. Our lower bound construction is a variant of the well-known ``hat graph'', which has been known since~\cite{babaioff2007matroids}  to be problematic for greedy-like algorithms. So our main contribution is not this construction itself, but rather a formalization of precisely the class of greedy algorithms for which this graph is problematic.

\begin{mainresult}[Informal, see Theorem~\ref{thm:greedy}]
No Greedy algorithm (as per Algorithm~\ref{algo:Greedy}) is constant-utility-competitive.
\end{mainresult}

We emphasize that while the hat graph itself is not a novel construction, our proof is quite distinct (and more involved) from prior work as it must rule out a broad class of algorithms rather than just a single one.

%---------------------------------------------------------------------------
\subsection{Randomized Partition Algorithms}

Another class of particularly simple algorithms are \emph{randomized partition algorithms}: 
\begin{enumerate} \itemsep0em 
    \item Before looking at any weights, (perhaps randomly) partition all the elements\footnote{We consider the known matroid setting where the matroid is known but the weights are revealed one-by-one.} into parts $S_i$.
    \item Within each part, run Dynkin's algorithm.
    \item Output the union of the selected elements.
\end{enumerate}
Note that these algorithms are allowed to use \emph{any randomized} partition. The elegant $1/(2e)$-approximation of Korula and Pal for graphic matroids\footnote{Given a graph with edges $E$, a graphic matroid $(E,\mathcal{I})$ is defined with $\mathcal{I}$ consisting of all subsets of edges that do not contain a cycle.} is a randomized partition algorithm~\cite{korula2009algorithms}. Their algorithm is utility-competitive, but not probability-competitive. Soto et al.~\cite{SotoTV-SODA18} recently designed a different constant probability-competitive algorithm for graphic matroids. While their algorithm is still quite elegant, it is perhaps not quite as simple as randomized partition algorithms. It is also worth noting that algorithms such as~\cite{lachish2014log, FSZ-SODA15} %(or~\cite{feldman2016online} for the related matroid prophet inequality) 
follow a more general framework, where  the algorithm in step one  looks at the weights before partitioning and step two is not necessarily Dynkin's single-choice algorithm (but perhaps some simple greedy algorithm).
This raises the question whether the novel development beyond~\cite{korula2009algorithms} is necessary to achieve probability-competitive algorithms? 
Our second main result answers this question: no randomized partition algorithm can be constant-probability-competitive (or even $\omega(n^{-1/8})$-probability-competitive).

\begin{mainresult}[Informal, see Theorem~\ref{thm:partition-neg}]
No Randomized Partition algorithm is constant-probability-competitive.
\end{mainresult}

Our construction witnessing Theorem~\ref{thm:partition-neg} is also a graphic matroid, although it is unrelated to the hat graph (and to the best of our knowledge, novel). Note that our proof cannot be extended to utility-competitive algorithms since we know~\cite{korula2009algorithms} is a constant-utility-competitive randomized partition algorithm for graphic matroids.

%---------------------------------------------------------------------------
\subsection{Related Work and Brief Summary} \label{sec:related}
There is a \emph{substantial} body of work on random-order problems for  matroids (the Matroid Secretary Problem~\cite{babaioff2007matroids}) and for several other discrete optimization problems;  we will not attempt to overview it  (e.g., see~\cite{GS-arXiv20,dinitz2013recent}). Here, we will briefly repeat the most related works.

Our work takes first steps towards characterizing classes of algorithms which might resolve the Matroid Secretary Problem. We focus on the simplest classes of algorithms which previously succeeded in special cases or for related problems, Greedy~\cite{JailletSZ13,kleinberg2012matroid} or Randomized Partition~\cite{korula2009algorithms}, and study the limits of these classes. 

First, we consider extremely simple greedy algorithms. A specific instantiation of this class of algorithms was shown to fail on a now-canonical ``hat graph'' in~\cite{babaioff2007matroids}, but related algorithms known to succeed in the Free-Order Model~\cite{JailletSZ13, AzarKW14}, and in the related Matroid Prophet Inequality~\cite{kleinberg2012matroid}. In addition, Dynkin's algorithm and the Optimistic algorithm for $k$-uniform matroids of~\cite{babaioff2007knapsack} fit this model. Our Theorem~\ref{thm:greedy} shows that no  Greedy algorithm is constant-utility-competitive for all matroids.

Second, we consider probability-competitive algorithms, formally considered in~\cite{SotoTV-SODA18}, and related to the ordinal model considered in~\cite{HoeferK-ICALP17}. Soto et al.~\cite{SotoTV-SODA18}, in particular, develop several probability-competitive algorithms for core settings such as graphic, transversal, and laminar matroids. Our work asks whether the extremely simple algorithms previously developed in~\cite{korula2009algorithms} can match these stronger probability-competitive guarantees, and we show in Theorem~\ref{thm:partition-neg} that the answer is no.

% !TeX root = main.tex
% !TEX root = main.tex

\section{Preliminaries\label{ch:prelim}}

The Matroid Secretary Problem (MSP) is defined as:
\begin{enumerate}\itemsep0em
\item There is a matroid $\mathcal{M}=(V,\mathcal{I})$, and weight function $w(\cdot):V \rightarrow \mathbb{R}_{\geq 0}$. Matroid $\mathcal{M}$ is fully-known to the algorithm in advance.\footnote{We are not concerned with computational efficiency of our algorithms in this work (our lower bounds are unconditional), so we will not stress about the precise format in which access to the matroid is given. To be concrete, one access model is that the algorithm has oracle access to $\mathcal{I}$ (query a  set $S$ and learn whether or not $S \in \mathcal{I}$). To the best of our knowledge, most algorithms previously considered for MSP are polytime given oracle access to $\mathcal{I}$.} Function $w(\cdot)$ is initially completely unknown to the algorithm.

    \item Initially, the set of accepted elements, $A$, is empty. Elements of $V$ arrive in a uniformly random order. When an element $i\in V$ arrives, the algorithm learns its weight $w(i)$, and must make an immediate and irrevocable decision whether or not to accept it (adding it to $A$). The algorithm must maintain  $A \in \mathcal{I}$ at all times.
\item If set $A$ is selected, the algorithm achieves payoff $\sum_{i \in A} w(i)$. 
\end{enumerate}

We will abuse notation and use $w(S):=\sum_{i \in S} w(i)$. Because $w(\cdot)$ is fixed, the offline optimum is the max-weight basis: $\text{MWB}(\mathcal{M}):=\arg\max_{S \in \mathcal{I}}\{w(S)\}$.\footnote{In this work, we assume for simplicity that the max-weight basis is unique. In case of ties, we tie-break by choosing the lexicographically-earlier basis.} We will also use standard matroid notation such as \emph{restriction}: the matroid $\mathcal{M}|_S$ is the matroid $\mathcal{M}$ restricted to $S$, and has ground set $S$ and independent sets $\mathcal{I}|_S:= \{T \cap S \mid T \in \mathcal{I}\}$. We also discuss matroid \emph{contractions}: the matroid $\mathcal{M} \setminus S$ is the matroid $\mathcal{M}$ contracted by $S$, and has ground set $V \setminus S$ and independent sets $\mathcal{I}\setminus S:=\{T \mid T \cup S \in \mathcal{I}\}$. When $\mathcal{M}$ is clear from context, we will also (slightly) abuse notation and write $\text{MWB}(T):=\text{MWB}(\mathcal{M}|_T)$.

%\subsubsection{Dynkin's Algorithm}

We will later reference Dynkin's $1/e$-probability-competitive algorithm for selecting a single item, i.e., a $1$-uniform matroid: (1)  Reject the first $T=\text{Binom}(n,1/e)$ elements and call this the \emph{sampling stage}. (2) Afterwards, accept an element $i$ iff it is the heaviest element seen so far. 

\begin{theorem}[\cite{dynkin1963optimum}]
Dynkin's algorithm is $1/e$-probability-competitive for $1$-uniform matroids, this is optimal.
\end{theorem}

%\subsection{Matroid Basics}
%In this section, we will provide several matroid-related definitions.

\IGNORE{
\begin{definition}[Bases]\label{def:basis}
Let $\mathcal{M}=(V,\mathcal{I})$ be a matroid. Then a set $B$ is called a \emph{basis} of $\mathcal{M}$ if $B$ is a maximal independent subset of $V$. In other words, $B$ is a basis iff $B\in\mathcal{I}$, and $B\cup\{x\}\not\in\mathcal{I}$ for all $x\in V \setminus B$.
\end{definition}
\noindent
By the exchange property, all bases of a matroid have the same size.

\begin{definition}[Max-weight basis (MWB)]\label{def:mwb}
Let $\mathcal{M}=(V,\mathcal{I})$ be a matroid, and consider a weight function $w(\cdot):V\rightarrow \mathbb{R}_{\geq 0}$. Furthermore, define the weight of a set $S\subseteq(V)$ as $w(S)=\sum_{e\in S} w(e)$. Then a basis $B$ is called a \emph{max-weight basis} if $w(B)\geq \max w(S)$, where the max is taken over all bases $S$.
\end{definition}
}

% !TeX root = main.tex
% !TEX root = main.tex

\section{Greedy Algorithms\label{ch:prophet}}
Because matroids are exactly the constraints for which the simple greedy algorithm is optimal, greedy-like algorithms are a natural family to consider as candidates for resolving the Matroid Secretary Problem. Indeed greedy-like algorithms solve the related Matroid Prophet Inequality~\cite{kleinberg2012matroid}, Matroid Secretary in the free-order model~\cite{JailletSZ13, AzarKW14}, and special cases of Matroid Secretary~\cite{dynkin1963optimum,babaioff2007knapsack}. In this section, we give an impossibility result for certain greedy algorithms. This helps unify counterexamples for related algorithms, and also helps narrow future research towards algorithms which have hope of resolving the Matroid Secretary Problem. 

\subsection{A Class of Greedy Algorithms}\label{sec:properties}
We now define a natural framework of greedy algorithms for the Matroid Secretary Problem (Algorithm~\ref{algo:Greedy}). Without loss of generality, we consider the continuous arrival setting, where each element $e \in V$ arrives at a time $t(e)$ independently and uniformly drawn from $[0,1]$. We refer by $V_t$ to the set of elements that arrive (strictly) before $t$, and by $A_t$ to the set of elements accepted by the algorithm (strictly) before time $t$.

\begin{algorithm}\caption{Greedy Algorithm for the matroid secretary problem}\label{algo:Greedy}
\begin{flushleft}
\quad We define a \emph{greedy algorithm}  as one that satisfies the following properties:
\end{flushleft}
\begin{enumerate}[(i)]\itemsep0em%[nosep]\enumerate
    \item \label{prop1} Reject (but store) elements that arrive before $T$ (sampling stage). Denote $S:=V_T$ to emphasize this.
    \item \label{prop3} At all times $t$, maintain an independent set $I_t$ such that:
    \begin{itemize}\itemsep0em%[nosep, label=-]
        \item $I_t$ contains all accepted elements and no elements which were rejected after $T$, i.e. $A_t\subseteq I_t\subseteq A_t\cup S$.
        \item At all times $t$, $I_t$ spans $V_t$.
    \end{itemize}
  \item \label{prop2} Accept $e$ if and only if $e\in \text{MWB}((\mathcal{M}\setminus A_{t(e)})|_{I_{t(e)}\cup \{e\}})$ (and $t(e) > T$). That is, accept $e$ if and only if it is in the max-weight basis of $I_{t(e)}\cup\{e\}$ \emph{after contracting by $A_{t(e)}$}.
\end{enumerate}
\end{algorithm}

Before getting into our results, it is helpful to understand why Algorithm~\ref{algo:Greedy} is a class of algorithms (rather than a fully-specified algorithm). The reason is that the algorithm has flexibility in which subset of $S$ to include in $I_t$ (but it must include $A_t$, and must span $V_t$). The restriction is that the algorithm does not know which element might arrive at time $t$, nor its weight, when setting $I_t$. Furthermore, the algorithm can choose the length of the sampling stage $T$.

%\filbreak
It is also helpful to see how this framework captures (or doesn't capture) existing greedy-like algorithms:
\begin{itemize}
\item Dynkin's algorithm (with $T = 1/e$) fits this framework. But so do suboptimal algorithms (e.g., accept the first element after $T$ which exceeds the $5^\text{th}$-highest sample. Or even accept an element which arrives at time $t > T$ iff it exceeds the $(\lfloor 5t/T\rfloor)^{\text{th}}$-highest sample).
\item The Optimistic Algorithm for $k$-uniform matroids of~\cite{babaioff2007knapsack} fits this framework. The algorithm maintains a list $U$, initially the $k$ heaviest elements of $S$. If $e$ exceeds the lightest element in $U$, it is accepted, and the lightest element of $U$ is removed. In our language, this has $I_t:=A_t \cup U$ at all times.
\item There is a natural extension of the Optimistic Algorithm to all matroids, which was previously considered in~\cite{babaioff2007matroids}, that we define as the \emph{supergreedy} algorithm shortly and analyze as a warmup in Section~\ref{sec:supergreedy}.
\item A related Pessimistic Algorithm (similar to the rehearsal algorithm for the related $k$-uniform prophet inequality of~\cite{AzarKW14}) for $k$-uniform matroids fits this framework. The algorithm also maintains a list $U$, initially the $k$ heaviest elements of $S$. If $e$ exceeds the lightest element in $U$, it is accepted, but the \emph{heaviest element of $U$ lighter than $e$} is removed. In our language, this again has $I_t:=A_t \cup U$ at all times (but $U$ is updated differently to the previous bullet).
\item The Virtual Algorithm for $k$-uniform matroids of~\cite{babaioff2007knapsack} does \emph{not} fit this framework. The algorithm accepts an element $e$ if and only if $e$ is one of the heaviest $k$ elements so far \emph{and} the $k^{\text{th}}$-heaviest element of $V_{t(e)}$ is in $S$ (i.e., $e$ is accepted if and only if it ``kicks out a sample'' from the top $k$ so far). This is because the algorithm needs to remember rejected elements in order to properly keep track of the $k^{\text{th}}$-heaviest element so far, and whether it was a sample.
\end{itemize}

Observe finally that all of the algorithms above (which fit the framework) further have the following. First, if an element is rejected (after $T$), it is forgotten forever, and the algorithm proceeds as if the element had never existed in the first place.\footnote{But, the framework is rich enough to also allow for algorithms which update $I_t$ as they reject an element. This makes impossibility results stronger.} Similarly, once an element $e$ is accepted, the algorithm updates $\mathcal{M}$ by contracting by $e$, and then proceeds identically as if the true matroid had been $\mathcal{M} \setminus \{e\}$ the whole time.\footnote{The framework is rich enough to allow for algorithms which update $I_t$ based on $A_t$, rather than just $\mathcal{M}\setminus\{A_t\}$, which again just makes impossibility results stronger.} These attributes are shared by the matroid prophet inequality of~\cite{kleinberg2012matroid}, and initially drove our formulation.

%\filbreak
With an understanding of Greedy algorithms in hand, we now state our main result.
\begin{theorem}\label{thm:greedy} Any algorithm satisfying the 3 properties of Algorithm~\ref{algo:Greedy} cannot be constant-utility-competitive. 
\end{theorem}
\subsection{Hard Instance: The Hat}\label{sec:hat}
In this section, we will study a \emph{hat graph} which drives our impossibility result. The hat has a special element which is significantly heavier than the sum of all others, and thus any algorithm with a good utility-competitive ratio must accept it. Furthermore, this special element appears in many small circuits, so the algorithm must not accept the remaining elements of any of these circuits prior to the arrival of the heavy element (otherwise, the heavy element cannot be accepted when it arrives). The hat was used in \cite{babaioff2007matroids} as a counterexample against a particular greedy algorithm (discussed in Section~\ref{sec:supergreedy}), and variants of the graph have been informally known to be problematic for ``greedy-like'' algorithms. However, prior to our work there was no formal classification of ``greedy-like''.

The hat on $n+2$ vertices is a collection of $n$ triangles, all sharing the same edge. Formally, an undirected graph $(V,E)$ is a hat if $V=\{a,b,v_1,\ldots,v_n\}$ for some $n>0$, and $E=\{\{a,b\}\}\cup\big\{e_i=\{a,v_i\} : i\in [n]\big\}\cup\big\{e'_i=\{b,v_i\} : i\in [n]\big\}$. Several weight assignments to the edges of the hat can serve as counterexamples to the algorithms considered in this section, but we consider a particular weight assignment for ease of exposition (as we only need one counterexample). We define this weight function $w: E\rightarrow \mathbb{R}_{\geq 0}$ to maintain the following ordering of the edge weights: $w(e_1) >\ldots> w(e_n) > w(e_1') >\ldots> w(e'_n)$. Furthermore, $w(\{a,b\})$ is much larger than the sum of the weights of all other edges. We will refer to $\{a,b\}$ as the \emph{infinity edge}, and we refer to its arrival time as $t_\infty:=t(\{a,b\})$ to emphasize this. Additionally, we consider the drawing of the hat in the plane as shown in Figure~\ref{fig:hat}, where $e_i$ is to the left of $e_j$ for $i < j$, and $e_i$ is above $e'_i$ for all $i$. Accordingly, we will sometimes refer to the relative position of edges to imply a relation between their relative weights.
\begin{figure}
    \centering
    \includegraphics[scale=0.75]{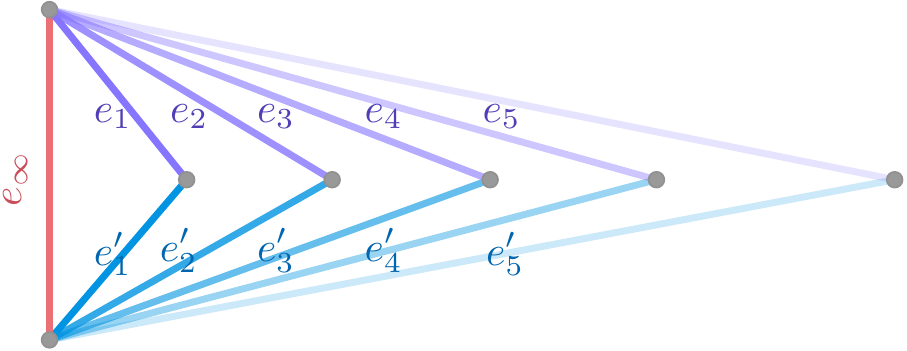}
    \caption{A hat on seven vertices. All purple edges ($e_1,\ldots,e_5$) are heavier than all blue edges  ($e'_1,\ldots,e'_5$), and $e_\infty$ is significantly heavier than all other edges. Within each color, darker edges are heavier.}
    \label{fig:hat}
\end{figure}

We call the pair of edges $(e_i,e'_i)$ the $i$-th \emph{claw}. Recall that any algorithm satisfying the 3 properties listed in Section~\ref{sec:properties} has memory limited to an independent set $I_t$. At any time $t$, given the history of arrivals and the algorithm's past decisions, we can classify the claws into one of 9 kinds in $\{-,A,S\}^2$. The first character in the pair describes the state of the top edge $e_i$, and the second character describes the state of the bottom edge $e'_i$. $S$ refers to an edge that is in $I_t$ and arrived in the sampling stage. $A$ refers to an edge that has been accepted by the algorithm (and is therefore in $I_t$). $-$ refers to any edge that is not in $I_t$. For example, if the $i$-th claw is of type $(S\,-)$ at some time $t$, it means that $t(e_i)<T$, $e_i\in I_t$, and $e'_i\not\in I_t$. Figure~\ref{fig:claws} illustrates these claws.

\begin{figure}
    \centering
    \includegraphics[scale=0.7]{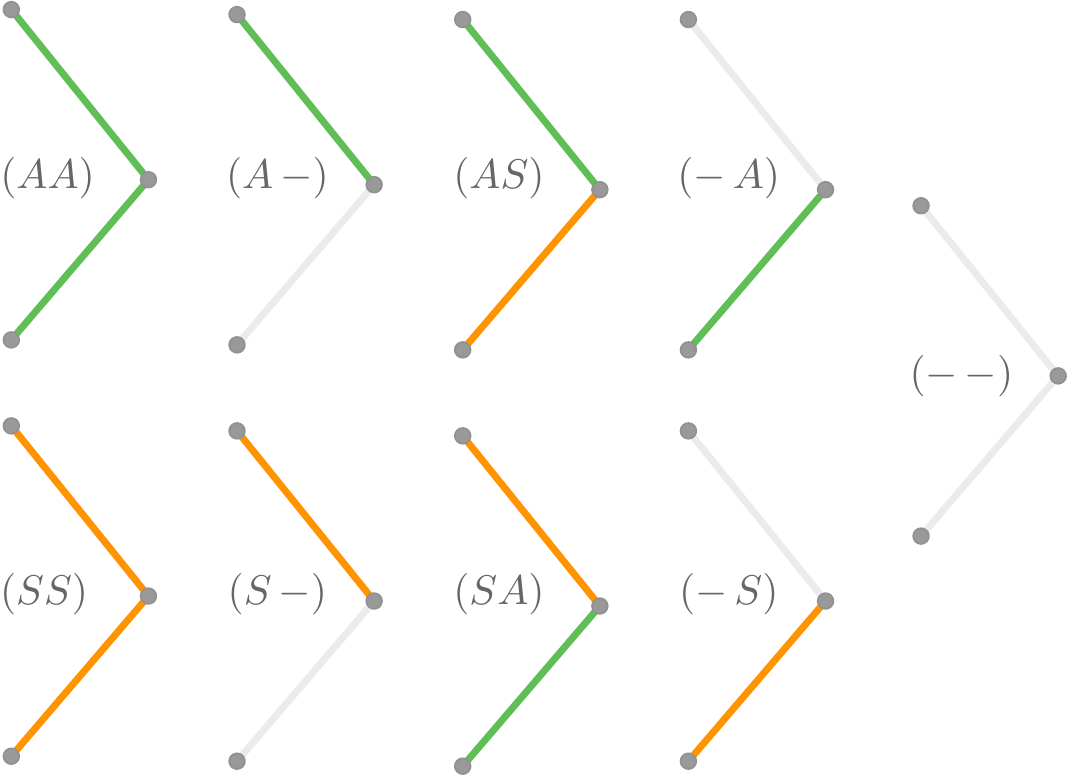}
    \caption{All possible kinds of claws at any time $t$. $S$ refers to sample edges in $I_t$ (drawn in orange), $A$ refers to an accepted edge in $I_t$ (drawn in green), and $-$ refers to any other edge (drawn in gray).}
    \label{fig:claws}
\end{figure}

We next state a few lemmas about different classes of claws and their implications about the performance of the algorithm. Since the infinity edge weighs significantly more than other edges combined, we say the algorithm ``loses'' ({i.e.}, fails to have a constant utility-competitive ratio) if it fails to accept the infinity edge. Conversely, the algorithm ``wins'' if it accepts the infinity edge. Our first observation characterizes the exact scenarios in which the algorithm loses. All missing proofs in this section can be found in Appendix~\ref{app:hat}.

\begin{observation}[Loss condition]\label{obs:aa-win}
The algorithm loses iff there is an $(AA)$ claw before $t_\infty$.\footnote{Babaioff~et~al.~\cite{babaioff2007matroids} used the same graph as a counterexample to a special case of our greedy algorithm (see Section~\ref{sec:supergreedy}), also relying on this observation. Our lemmas are otherwise new, and necessary since we rule out a much larger class of greedy-like algorithms.} 
\end{observation}

The next lemma specifies the unique blocking structure that would prevent the loss-inducing $(AA)$ claws from forming. Our analysis focuses on the case of a $(-\,A)$ claw becoming a $(AA)$ claw, as these events are significantly more likely than a $(A\,-)$ claw turning into an $(AA)$ claw, and suffice for our analysis.

\begin{lemma}[Blockers and Protection]\label{obs:blocker}
Suppose there is no $(AA)$ claw yet. Consider a $(-\,A)$ claw whose upper edge is about to arrive. The upper edge is accepted iff there is no $(SA)$ claw to its left. For this reason, we will refer to $(SA)$ as the \emph{blocker}. We say that the algorithm is \emph{protected} at time $t$ if there is a blocker in $I_t$.
\end{lemma}

Importantly, note that there can be \emph{at most one blocker in $I_t$}, as two blockers form a cycle. So we can unambiguously refer to \emph{the} blocker at any time $t$. A blocker's effectiveness is a function of its location: Blockers far to the left ``protect'' more claws and are therefore more effective. 

With this language in mind, we can reframe the algorithm's objective, while working within the Greedy framework. The algorithm loses whenever the upper edge of a $(-\,A)$ claw arrives without a blocker to its left. So the algorithm would like to maintain a blocker in $I_t$ as far to the left as possible.\footnote{Note that an arbitrary algorithm can simply decide to violate the properties defining Greedy. Our goal is to analyze Greedy algorithms, which must fit this framework.} So the remainder of this section studies decisions the algorithm can make (again, within the Greedy framework) to include blockers far to the left. Lemma~\ref{obs:unprotected}, however, establishes that we cannot create a new blocker without destroying our old one first (thereby going ``unprotected'' for some period).

\begin{lemma}\label{obs:unprotected}
If the lower edge of an $(S\,-)$ arrives at time $t$ and $I_t$ has a blocker, this edge will not be accepted.
\end{lemma}

Lemma~\ref{obs:unprotected} means that the algorithm faces a tradeoff. If $I_t$ has a blocker, it is safe from accepting the upper edge of a $(-\,A)$ claw \emph{to its right} at time $t$. But, the algorithm \emph{cannot move its blocker to the left, even if the lower edge of an $(S\,-)$ arrives during this interval}. Alternatively, the algorithm may not have a blocker during $I_t$. In that case, the algorithm can possibly accept a good blocker, if one happens to arrive at time $t$. But, the algorithm is at risk of accepting the upper edge of a $(-\,A)$ claw that arrives at time $t$ \emph{no matter its location}, because $I_t$ has no blockers at all.

 \subsection{Warm Up: Ruling out the Supergreedy Algorithm}\label{sec:supergreedy}
In the appendix of the same paper in which they introduced the general matroid secretary conjecture, Babaioff~{et al.}~\cite{babaioff2007matroids} also show that the following supergreedy algorithm cannot be constant-utility-competitive. In Appendix~\ref{app:supergreedy} we will restate their proof in the slightly different language of claws and blockers, in part for completeness, and in part as a warm-up for the more involved proof in the following section. We first define the supergreedy algorithm,\footnote{We call this algorithm supergreedy, since in addition to the greedy property of always accepting elements according to the Greedy rule, it makes greedy choices about what elements to kick out of $I_t$ (in particular, ones that improve $I_t$ \emph{the most}).} which specifies a particular choice of $I_t$:

\begin{enumerate}[(i)]\addtocounter{enumi}{3}
    \item \label{prop6} At all $t$, $I_t$ is the maximum-weight set subject to constraints~\eqref{prop1},~\eqref{prop3},~\eqref{prop2} ($I_t:=\text{MWB}( (\mathcal{M}\setminus A_t)|_{V_T}) \cup A_t$).
\end{enumerate}

\begin{observation} The unique algorithm that satisfies properties \eqref{prop1} to \eqref{prop6} is Algorithm~\ref{algo:supergreedy}.
\end{observation}

\begin{algorithm}\caption{Supergreedy Algorithm}\label{algo:supergreedy}
\begin{flushleft} \quad Accept an element $e$ iff \end{flushleft}
\begin{enumerate}  \itemsep0em%[nosep, label=-]
    \item $e$ arrives after $T$ ($t(e) > T$);
    \item it is feasible to accept $e$ ($A_{t(e)} \cup \{e\} \in \mathcal{I}$);
    \item even after contracting by the accepted elements so far, $e$ is still in the max-weight basis of elements seen so far ($e \in \text{MWB}((\mathcal{M}\setminus A_{t(e)})|_{V_{t(e)}\cup \{e\}})$).
\end{enumerate}
\end{algorithm}

\begin{theorem}[\cite{babaioff2007matroids}]\label{thm:supergreedy}
The supergreedy algorithm is not $\alpha$-utility-competitive for constant $\alpha$.
\end{theorem}

A full proof (in our language) appears in Appendix~\ref{app:supergreedy}. But we overview the main idea here, as it will help demonstrate the difference between this particular greedy algorithm and an arbitrary one. Observe first that an $(SS)$ claw in $I$ prevents any blockers forming to its right.\footnote{This is because when the lower edge arrives, it is lighter than both edges in the $(SS)$ claw to its left, as well as the $S$ edge above it.} Observe also it is extremely likely that the left-most $(SS)$ claw is very far to the left, and $I_T$ must contain it. Therefore, in order to get a blocker, that blocker must either be to the left of the $(SS)$ claw (and thus will take a long time to arrive), or that $(SS)$ claw must be discarded (i.e. one of its edges must be removed from $I$).

The supergreedy algorithm does not discard the $(SS)$ claw, except with an $(AS)$ or $(SA)$ claw to its left (both of which take a long time to arrive). This means that until one of these claws arrive, we cannot block an $(AA)$ from occurring, and we are extremely likely to see an $(AA)$ (in fact, we're likely to see $\Omega(n)$ of them) before this happens.

The full proof just elaborates on the steps in this outline and makes calculations rigorous, but the outline above suffices to draw a distinction to arbitrary greedy algorithms. The supergreedy algorithm will only discard an $(SS)$ claw when the lower $S$ is no longer in the max-weight basis (contracted by $A$). In order for this to happen, we must have a complete $(AS)$ or $(SA)$ claw to its left, and this takes time to arrive. An arbitrary greedy algorithm, however, can instead immediately replace this far-left $(SS)$ claw with a far-right $(SS)$ claw, because it is not constrained to maintain a max-weight basis in $I_t$. Indeed, this flexibility allows the algorithm to actually engage in the tradeoff highlighted at the end of the previous section, and turn blockers `on' or `off'.

 \subsection{Main Result: Ruling out all Greedy Algorithms}\label{sec:all-greedy}
Armed with a better understanding of some properties of the hat structure and Theorem~\ref{thm:supergreedy} as warmup, we are ready to prove Theorem~\ref{thm:greedy}, which states that greedy algorithms fail to be $\alpha$-utility-competitive for any constant $\alpha$.

We give a detailed proof sketch below, and defer calculations to Appendix~\ref{app:all-greedy}. We first repeat the main intuition: The algorithm's goal is to not accept any $(AA)$ claw before $t_\infty$ (Observation~\ref{obs:aa-win}). To do so, the algorithm \emph{must} make sure $I_t$ includes a blocker to the left of every $(-\,A)$ whose upper edge arrives at time $t<t_\infty$ (Lemma~\ref{obs:blocker}). We can order potential blockers $(S\,-)$ by the arrival times of their lower edges, each of which is uniformly distributed in $[T,1]$. Therefore, it is unlikely that a blocker far to the left arrives very early. 

The algorithm can try to start with a mediocre blocker and improve it over time by accepting blockers further to the left as they arrive. The caveat is that due to Lemma~\ref{obs:unprotected}, \emph{blocker improvements are only possible in unprotected periods}, during which \emph{any} arriving upper edge of $(-\,A)$ claws is accepted. Therefore, the algorithm faces a trade-off: Forming a more effective blocker costs more unprotected time. Importantly, the algorithm does not know whether the next arriving edge will be part of a potential blocker, or part of an $(-\,A)$. 

In order to show that the algorithm fails, we show that with high probability there will be a $(AA)$ claw before the arrival of the infinity edge. Specifically, we show that with high probability, an $(-\,A)$ claw becomes $(AA)$ in an interval of length $\ell = n^{-0.1}$ after $T$, which is with high probability before the arrival of the infinity edge. 

We now get into details of our proof approach. We first choose a parameter $x \in [n]$ (thinking of the claws as labeled $1$ through $n$ from left to right). We will undercount the algorithm's failure, noting that it fails whenever any of the following happens:
\begin{itemize}
\item The upper edge of some $(-\,A)$ to the left of $x$ arrives during $[T,T+\ell]$, \emph{and} $I_t$ does not include any blocker to the left of $x$ for any $t\in[T,T+\ell]$.
\item The upper edge of some $(-\,A)$ arrives at an unprotected  $t\in[T,T+\ell]$.
\end{itemize}

In other words, we are zeroing in on two potential sources of failure: the upper edge of \emph{any} $(-\,A)$ claw could arrive during an unprotected time, or the upper edge of an $(-\,A)$ claw to the left of $x$ could arrive before the algorithm accepts a blocker to the left of $x$. Note that these are very narrow possibilities for failure, but they suffice for our analysis. 

So there are three probabilities to analyze. The first part of the first bullet is independent of the algorithm,\footnote{Recall that the first edge of a $(-\,-)$ claw to arrive must always be accepted since $I_t$ must span $V_t$.} and simply considers the probability that the upper edge of a $(-\,A)$ to the left of $x$ arrives during $[T,T+\ell]$.

\begin{lemma}\label{lem:AA} With probability at least $1-2^{-\ell^2x/2}$, the upper edge of a $(-\,A)$ claw to the left of $x$ arrives between $T$ and $T+\ell$.
\end{lemma}

The next two probabilities are significantly more involved, as they consider decisions made by the algorithm. Note that the algorithm can decide \emph{adaptively} when to go unprotected, based on the current ratio of $(-\,A)$s (potential $(AA)$s) versus $(S\,-)$s (potential blockers) to the left of $x$. To this end, we will let the algorithm adaptively choose any (measurable) subset of $[T, T+\ell]$ to go unprotected, and let $y$ denote the total measure of this interval.\footnote{The algorithm does not need to commit to the value of $y$ in advance or choose it deterministically.} $y$ captures the aforementioned tradeoff: small $y$ means that the algorithm is likely to fail bullet one, while large $y$ means the algorithm is likely to fail bullet two. Lemma~\ref{lem:bullettwo} quantifies the cost of keeping $y$ small, lowerbounding the probability of the second part of the first bullet.

\begin{lemma}\label{lem:bullettwo} Conditioned on the upper edge of a $(-\,A)$ claw to the left of $x$ arriving between $T$ and $T+\ell$ (i.e. Lemma~\ref{lem:AA} happening), any greedy algorithm which goes unprotected for a total measure of $y$ during $[T, T+\ell]$ fails to accept a blocker to the left of $x$ with probability at least:
\begin{align*}
(1-2x\ell e^{-\frac{2x}{3}})(1-y)^{4x}.
\end{align*}
\end{lemma}

Finally, we analyze the second bullet, lower bounding the probability that the upper edge of a $(-\,A)$ claw (anywhere) arrives during a period when the algorithm is unprotected (while the precise form is complicated, recall the intuition that as $y$ gets larger, the probability of this particular bad event goes up, and $y$ is at most $\ell$):

\begin{lemma}\label{lem:bulletone} Any greedy algorithm which goes unprotected for a total measure of $y\geq n^{-0.4}/2$ during $[T, T+\ell]$ has the upper edge of a $(-\,A)$ claw arrive during an unprotected $t$ with probability at least:
\[1-\left(1-\frac{2y-n^{-0.4}}{2\ell}\right)^{\frac{n^{0.6}(4\ell-n^{-0.4})}{32\ell^2}}.\]
\end{lemma}

Finally, we just need to combine the three bounds in Lemmas~\ref{lem:AA},~\ref{lem:bullettwo},~\ref{lem:bulletone}. We will choose a value of $\ell$ and $x$ for the analysis, and then the algorithm (knowing $x$) can adaptively allocate the unprotected intervals within $[T,T+\ell]$ for a total measure of $y$.
More formally, we let $f(y) = (1-2x\ell e^{-\frac{2x}{3}})(1-y)^{4x}\left(1-(\frac{1}{2})^{\ell^2x/2}\right)$ denote the lowerbound on failure probability derived in Lemma~\ref{lem:bullettwo}. Furthermore, we let
\[ g(y) =
    \begin{cases} 
      1-\left(1-\frac{2y-n^{-0.4}}{2\ell}\right)^{\frac{n^{0.6}(4\ell-n^{-0.4})}{32\ell^2}}, & y \geq \frac{n^{-0.4}}{2}; \\
       0, & y< \frac{n^{-0.4}}{2}. 
   \end{cases}\]
The first case follows from Lemma~\ref{lem:bulletone}, and setting $g$ to $0$ elsewhere only strengthens our lower bound.
Overall, the algorithm fails with probability at least $\min_y\{\max\{ f(y),g(y) \} \}$. The next lemma sets parameters to lower bound this expression.

\begin{lemma}\label{lem:combine} When $x = n^{0.3}$ and $\ell=n^{-0.1}$, we have
\[\lim_{n \rightarrow \infty}\min_{y \in [0,\ell]} \max \{f(y),g(y)\} = 1.\]
\end{lemma}

The proof of Theorem~\ref{thm:greedy} now follows from the four lemmas of this section.

% !TeX root = main.tex
% !TEX root = main.tex

\section{Randomized Partition Algorithms\label{ch:partition}}
This section is devoted to a class of algorithms based on partition matroids.
These are generalizations of an algorithm by Korula and Pal~\cite{korula2009algorithms} for the secretary problem on graphic matroids. We begin by briefly describing their result in Section~\ref{sec:korula-pal}, proceeding to show in Section~\ref{sec:negative} and Section~\ref{sec:negativeRandomized} that this algorithm and several natural generalizations of it fail to provide good probability-competitive performance.

\subsection{The Korula-Pal~\cite{korula2009algorithms} Algorithm for Graphic Matroids}\label{sec:korula-pal}
Recall that for any undirected connected graph $G=(E,V)$, one can define a graphic matroid $\mathcal{M}=(E,\mathcal{I})$, where the ground set $E$ is exactly the edge set of $G$ and the independent set collection $\mathcal{I}$ is the set of all acyclic subgraphs of $G$. 
The bases of a graphic matroid are exactly the set of spanning trees\footnote{Or spanning forests in the case of disconnected graphs.} of the graph. 

The following algorithm due to Korula and Pal~\cite{korula2009algorithms} achieves a utility-competitive ratio of $2e$ on graphic matroids:
\begin{algorithm} \caption{Korula-Pal~\cite{korula2009algorithms} Algorithm} \label{alg:KorulaPal}
\begin{enumerate}\itemsep0em%[nosep]
    \item Before any edges arrive, partition the edges of the graph offline as follows:
    \begin{itemize} \itemsep0em%[nosep]
        \item Pick a random ordering $\sigma: V\rightarrow [n]$ of nodes.
        \item For each edge $e = (u,v)$, with $\sigma(u) < \sigma(v)$, put $e$ in $E_u$ (otherwise, put it in $E_v$).
    \end{itemize}
    \item Run Dynkin's algorithm on all the $E_v$s in parallel, and  output the union of all selected edges.
\end{enumerate}
\end{algorithm}
Figure~\ref{fig:korula-pal} shows an instance of such a random partitions.

\begin{figure}
\includegraphics[scale=.7]{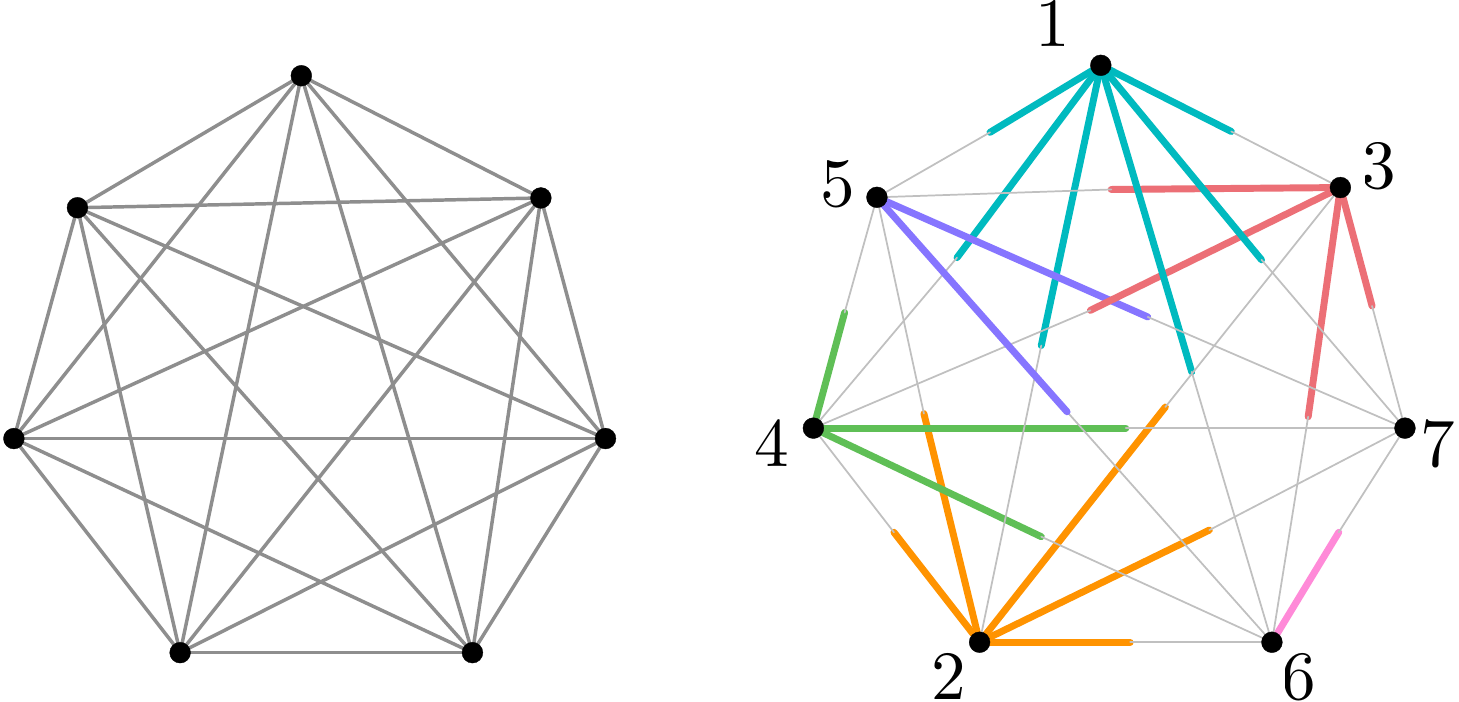}
\centering
\caption{It shows  a partition of edges of a complete graph where edges in the same part are drawn with the same color. The numbers next to each vertex denote the  $\sigma$  values.}\label{fig:korula-pal}
\end{figure}

\begin{theorem}[\cite{korula2009algorithms}] \label{thm:KorulaPal}
Algorithm~\ref{alg:KorulaPal} is $1/(2e)$-utility-competitive.
\end{theorem}

For completeness, we present the (short) proof of this theorem in Appendix~\ref{sec:ProofKorulaPal}, and highlight why the proof gives utility-competitiveness but not probability-competitiveness.

%--------------------------------------------------------------------------

\subsection{Defining Randomized Partition Algorithms}
The algorithm by Korula and Pal~\cite{korula2009algorithms} was phrased in the language of graphs. Let us try to generalize it in a language applicable to all matroids. Before seeing any weights, their algorithm restricts itself (randomly) to accepting only a subset of independent sets. More specifically, the algorithm will restrict its attention to the disjoint union\footnote{This disjointness is why we refer to these generalizations as algorithms based on ``partition matroids.''} of solutions to simpler subproblems. The algorithm must ensure that for all feasible solutions to the subproblems, their union is a feasible solution to the main problem. In the case of the Korula-Pal algorithm, the smaller subproblems are instances of 1-uniform matroid secretary problems. (Several other algorithms for the Matroid Secretary Problem use similar high-level techniques, where the ``simpler'' matroids are not 1-uniform~\cite{FSZ-SODA15, huynh2016matroid, lachish2014log, soto2011simple}, and this idea is also used for the related prophet inequality~\cite{feldman2016online}.)

More concretely, we say that a partition is \emph{valid} if the union of what is accepted by the instances of Dynkin's algorithm is an independent set (regardless of the weights and order of arrivals). Now we consider the following class of algorithms based on partition matroids:

\begin{algorithm}\caption{Randomized Partition} \label{alg:randomizedPartition}
\begin{enumerate}\itemsep0em%[nosep]
    \item Before looking at any weights, (perhaps randomly) validly partition the elements into parts $S_i$.
    \item Within each part, run Dynkin's algorithm, and output the union of the selected elements.
\end{enumerate}
\end{algorithm}
\noindent One can ask whether any algorithm in this framework can be constant-\emph{probability}-competitive. Theorem~\ref{thm:partition-neg} shows that the answer is `no'.

\subsection{Warm Up: Deterministic Partitions}\label{sec:negative}
We show as a warm-up that no fixed \emph{deterministic} partition can achieve a constant probability-competitive ratio.
 Consider the following algorithm based on partition matroids.

\begin{algorithm} \caption{Deterministic Partition}\label{algo:detPartition}
\begin{enumerate}\itemsep0em
\item Deterministically fix a valid partition $\{S_1,\ldots,S_k\}$ of all elements of the ground set offline. 
\item Run Dynkin's algorithm on each part online, and output the union of the selected edges.
\end{enumerate}
\end{algorithm}

\begin{prop}\label{prop:deterministic}
Algorithm~\ref{algo:detPartition} is not $\alpha$-probability-competitive for $\alpha>\frac{2\sqrt{2}}{\sqrt{n}}$.
\end{prop}
\begin{proof}
We will give a counter-example for graphic matroids. Fix a partition $\{S_1,\ldots,S_k\}$ of the edges of the complete graph $K_n$. We will construct a weight function $w$ for which the algorithm performs  poorly.

Without loss of generality, assume all the $S_i$ are non-empty. We claim that $k\leq  n-1$. Otherwise, there would be a set $A$ of at least $n$ edges, such that each edge comes from a distinct part. By validity of the partition, this set of edges must be acyclic. However, any set of $n$ edges on $n$ vertices must have a cycle, which is a contradiction.

Since there are ${n\choose 2}$ edges distributed into at most $n-1$ parts, there must be some part $S_i$ that contains at least $n/2$ edges. Consider the graph $G_i$ defined by the edges in $S_i$ and the vertices they touch. Then $G$ has at least $\sqrt{\frac{n}{2}}$ vertices, each of which has degree at least 1. Therefore, $G$ has a spanning forest with at least $\sqrt{\frac{n}{2}}/2$ edges. 

Consider a weight function $w$ that assigns a weight of 1 to all edges in this forest, and a weight of 0 to every other edge. Clearly, the forest is the max-weight basis of the graphic matroid. However, Dynkin's can only choose one out of at least $\sqrt{\frac{n}{2}}/2$ edges from the max-weight basis, so each edge cannot be chosen with probability higher than $\frac{2\sqrt{2}}{\sqrt{n}}$.
\end{proof}

\subsection{Randomized Partitions}\label{sec:negativeRandomized}
In this section, we will rule out all algorithms based on partition matroids as candidates for achieving a constant probability-competitive ratio for the matroid secretary problem. As opposed to  the algorithm ruled out in Proposition~\ref{prop:deterministic} (which used any deterministic partition),  these algorithms are allowed to use \emph{any randomized} partition.

For the algorithm to always output a feasible solution, any partition it uses must be valid. Recall that a \emph{valid} partition is one for which the union of what is accepted by the instances of Dynkin's algorithm is always independent. We say a distribution over partitions is \emph{valid} if every partition in its support is valid.

Without loss of generality, we can assume the input graph is always complete. Otherwise, one can consider a modified weight-function that assigns a weight of zero to every edge that is not present. Since the algorithm cannot see the weights of the edges in advance, it will have to choose a partition of the complete graph at the start.

\begin{theorem}\label{thm:partition-neg}
Any algorithm that draws a partition from a valid distribution $\mathcal{D}$  in Algorithm~\ref{alg:randomizedPartition} is not $\alpha$-probability-competitive for any $\alpha=\omega(n^{-1/8})$.
\end{theorem}

%---------------------------------------------------------

The high-level plan in the proof of Theorem~\ref{thm:partition-neg} is to plant a random \emph{broom}, illustrated in Figure~\ref{fig:broom}, and show that with high probability, its handle is not accepted.  We will refer to the lone neutral edge $\{u,w\}$ connecting the two stars as the \emph{handle} of the broom. Note that the edges of non-zero weight in the broom form an acyclic subgraph and are therefore the unique max-weight basis of this graphic matroid. 

\begin{figure}
    \centering
    \includegraphics[scale=0.75]{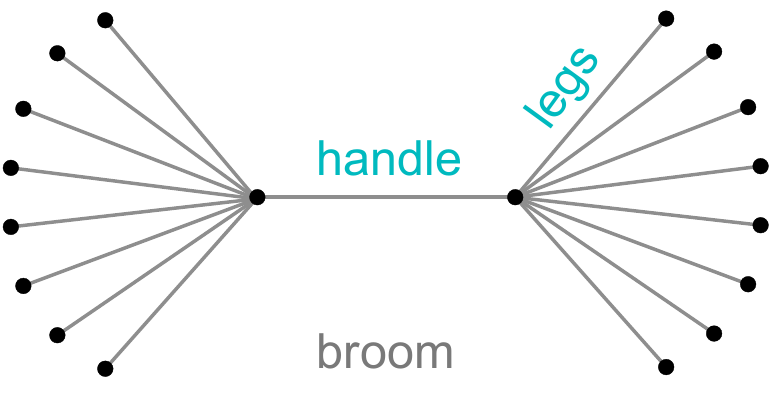}
    \caption{Two stars connected by an edge form a broom. We call the bridge between the two stars the handle of the broom, and we  the other edges of the broom as its legs. }
    \label{fig:broom}
\end{figure}

Before proving this theorem, we characterize valid partitions.

%--------------------------------------------------------------

\subsubsection{Characterizing Valid Partitions}
In this section we give a few characterizations of what valid partitions look like, which serve to provide intuition into why validity is a strong enough condition that prevents partition-based algorithms from probability-competitiveness.

Similarly to the previous section, we define a \emph{valid} partition to be one where the union of what is accepted by the instances of Dynkin's algorithm is always an independent set, even for adversarial weights and arrival orders. We first give several equivalent descriptions of what valid partitions should look like in the case of graphic matroids, which provides certain structural properties enforced by validity. It will be later used to prove our  Theorem~\ref{thm:partition-neg}.

\begin{lemma}\label{thm:valid-char}
Let $\{S_1, \ldots, S_k\}$ partition the edges of a complete graph $K_n$, and let $\textrm{part}(e)$ denote the $S_i$ containing edge $e$. The following are equivalent:
\begin{enumerate}[(a)] \itemsep0em
    \item \textbf{Matroid condition:} $\{S_1, \ldots, S_k\}$ is valid.
    \item \textbf{Graph condition (i):} Every \emph{cycle} has at least two edges in the same part.
    \item  \textbf{Graph condition (ii):} Every \emph{triangle} has at least two edges in the same part.
\IGNORE{  
  \item \textbf{Counting condition:} Let $N_\angle$ denote the number of (not necessarily induced) paths of length 3 (2 edges) within the same part (which we will call \emph{claws}), {i.e.},
    $$N_\angle=\Bigl\lvert\{\{x,y,z\}\subset V: \textrm{part}(\{x,y\})=\textrm{part}(\{y,z\})\not=\textrm{part}(\{x,z\})\}\Bigr\rvert,$$
    and let $N_\Delta$ denote the number of triangles within the same part, {i.e.},
    $$N_\Delta=\Bigl\lvert\{\{x,y,z\}\subset V: \textrm{part}(\{x,y\})=\textrm{part}(\{y,z\})=\textrm{part}(\{x,z\})\}\Bigr\rvert.$$
    Then the following equality holds:
    $$ N_\angle - 2 N_\Delta = {n \choose 3}\quad\quad \emph{[Triangle Equality]}$$
    \item \textbf{Algorithmic condition:} $\{S_1, \ldots, S_k\}$ can be constructed from $\{V\}$ by repeated application of the following operation (called a \emph{validity-preserving} operation): Given a partition $T = \{T_i\}$ of the edges of $K_n$, construct a new partition of the edges of $K_n$ by further splitting some part $T_i$ into two parts $T'_i$ and $T''_i$ in a way that no 3 vertices that originally had exactly two edges among them in $T_i$ end up having both of those edges either in $T'_i$ or in $T''_i$ (In other words, the operation does not break any claws).
    }
\end{enumerate}
\end{lemma}

\begin{proof}
$(a)\Leftrightarrow(b)$ holds trivially. If $\{S_1, \ldots, S_k\}$ is valid, it should not be possible to construct a cycle by taking each edge from distinct parts. Conversely, if all cycles have at least two edges from the same part, no algorithm can accept a cyclic subgraph.

$(b)\Leftrightarrow(c)$: $(b) \implies (c)$ holds by definition, so it suffices to show $(c)\implies (b)$. Let $\mathcal{S}=\{S_1, \ldots, S_k\}$ be a partition satisfying (c). We say that a cycle $C$ is \emph{shattered} by $\mathcal{S}$ if each of its edges belongs to a different part in $\mathcal{S}$. Then (c) implies that $\mathcal{S}$ shatters no cycles of size 3. Assume for contradiction, that $\mathcal{S}$ shatters some cycle, and let $C=\{e_1,\cdots,e_\ell\}$ be the smallest cycle shattered by $\mathcal{S}$, where $\ell > 3$. Without loss of generality, assume $e_i\in S_i$ for $i\in[\ell]$. Let $e_1=\{u,v\}$ and $e_2=\{v,w\}$. Consider the chord $\{u,w\}$. By (c), the cycle on vertices $\{u,v,w\}$ is not shattered by $\mathcal{S}$, so we must have either $\{u,w\}\in S_1$ or $\{u,w\}\in S_2$. In either case, the cycle $C=\{\{u,w\}, e_3\cdots,e_\ell\}$ is shattered by $\mathcal{S}$ but has size smaller than $\ell$, which is a contradiction.
\end{proof}

%----------------------------------------------------------

\subsubsection{Proof of Theorem~\ref{thm:partition-neg}}

We  provide a counterexample in the case of graphic matroids using the broom. 
Consider a partition $\mathcal{S}=\{S_1,\ldots,S_k\}$ of the edges of the complete graph. We say an edge $e\in S_i$ is ``high-degree'' if the sum of the degrees of its endpoints within the same part $S_i$ is large. More concretely, we define the part-$i$ degree of a vertex $v$ as $\deg_i(v)=\lvert\{e=\{a,b\}\in S_i : v\in\{a,b\}\}\rvert$. Given an edge $e=\{a,b\}$ in part $S_i$, its degree is given by $\deg(e) = \deg_i(a) + \deg_i(b) - 1$, which intuitively means that we are counting all the incident edges in that part and the edge itself. An edge $e$ is said to be \emph{high-degree} if $\deg(e)\geq C$ for some $C$ that we will choose later.

We will show that a $1-o(1)$ fraction of the edges are high-degree for super-constant $C$. Therefore, an adversary can plant a random broom  by
assigning weights according to the following distribution: Pick a random edge $\{u,v\}$ in the graph, and randomly partition the vertices $V\setminus \{u,v\}$ into two parts $X$ and $Y$ of equal size (we assume $\lvert V \rvert$ is even). Assign a weight of 1 to every edge $\{u,x\}$ and $\{v,y\}$ for all $x\in X, y\in Y$, and a weight of zero to everything else. We will show that no matter what partition an algorithm chooses, the random edge $\{u,v\}$ will have a high-degree with high probability. The algorithm must therefore choose at most one edge from at least $C$ elements of $\OPT$. Hence, it cannot be better than $1/C$-probability-competitive.

It remains to show that a $1-o(1)$ fraction of the edges are high-degree for some super-constant $C$ in any valid partition $\mathcal{S}$ of the edges of the complete graph. A partition of the edges of $K_n$ can be thought of as a coloring of its adjacency matrix $A \in M^{n\times n}$ (ignoring diagonal entries) in the obvious way ({i.e.}, assign a different color to each $S_i$, and the color $\textrm{part}(e)$ to the entry of $A$ corresponding to $e$). In this notation, an entry of $A$ is low-degree if there are fewer than $C$ entries of the same color in its row or column. Note that by Lemma~\ref{thm:valid-char},  a partition is valid iff every triangle has at least two edges in the same part. In the matrix language, a partition is valid iff for every three row indices $u,v,w$, at least two of $A(u,v)$, $A(u,w)$ and $A(v,w)$ are the same color. We will show using this interpretation of feasibility that each row and column must mostly consist of high-degree entries. More specifically, we will fix a vertex $v$, and consider any other two vertices $u$ and $w$. 

\begin{proposition}\label{prop:recurrence} Let $C \leq (n-1)/2$ and let $T(n)$ be the maximum possible number of low-degree edges in any valid coloring of the complete graph on $n$ vertices. For any vertex $v$, let $x_i$ denote the number of edges adjacent to $v$ in partition $i$. Then:
\begin{align*}
    T(n) \leq  \max_{\vec{x}\in \mathbb{N}_{\geq 0}^{n-1}, \sum_i x_i = n, } \min\bigg\{ &\sum_i T(x_i) + 2C(n-1)~~,~~ T(n-1) + 2(n-1-\max_i\{x_i\}) \bigg\}.
\end{align*}

\end{proposition}
\begin{proof} There are two steps: for any $\vec{x}$, we show that both the left term and the right term are always upper bounds (and therefore their minimum is a valid upper bound too).

Intuitively, the left term is better when $\max_i\{x_i\}$ is not too large. To see that the left term is always an upper bound, consider the following cases. Below, let $X_i$ denote the set of nodes $z$ such that $(z,v)$ is in partition $i$ (and therefore $x_i:=|X_i|$). 
\begin{itemize}
\item First, consider each $X_i$, and consider the induced subgraph on just these $x_i$ nodes. The number of low-degree edges \emph{just counting those between two nodes in $X_i$} is at most $T(x_i)$, by definition of $T(\cdot)$. Clearly, a node must be low-degree in the induced subgraph to possibly be low-degree in the full graph. This means there are at most $\sum_i T(x_i)$ low-degree edges between two nodes in the same $X_i$.
\item Next, consider an edge between two nodes $x,y$ both $\neq v$ which are \emph{not} in the same $X_i$. This means that the edges $(v,x)$ and $(v,y)$ are \emph{not} colored the same, and therefore the edge $(x,y)$ \emph{must} share a color with one of them for $A$ to be valid. Whichever edge shares its color, we will charge its non-$v$ endpoint (e.g. if $(x,y)$ shares a color with $(v,x)$, we charge $x$). Observe that once a vertex is charged $C$ times, this means there are $C+1$ edges adjacent to it which share the color of $(v,x)$. This means that none of these edges are low-degree. Therefore, an edge can be low-degree only if its non-$v$ endpoint is charged at most $C$ times, and therefore there can be at most $C(n-1)$ such low-degree edges.
\item Finally, consider an edge adjacent to $v$. We will lazily upper bound the number of low-degree edges by just the total number of edges, $n-1$, and further upper bound it by $C(n-1)$ for cleanliness of the expression.
\end{itemize}
This establishes the left term, which holds for any $\vec{x}$. Now we establish the right term. Intuitively, the right term is a better bound whenever $\max_i\{x_i\}$ is large. Let $j:=\arg\max_i\{x_i\}$. If $x_i > C$, then there can be no low-degree edges adjacent to $v$ in $X_1$. Therefore, there are at most $(n-1-x_j)$ low-degree edges adjacent to $v$. On the subgraph induced by the $n - 1$ nodes other than $v$, there are clearly at most $T(n-1)$ low-degree edges by definition of $T(\cdot)$, and again any edge which is low-degree in the full graph must be low-degree in every induced subgraph. On the other hand, if $x_j \leq C$, then perhaps all edges adjacent to $v$ are low-degree, and we can only use this technique to give an upper bound of $T(n-1)+n-1$. In both cases, our bound is at most $T(n-1)+2(n-1-\max_i\{x_i\})$ as long as $C \leq (n-1)/2$. 
\end{proof}

We will show inductively in Lemma~\ref{lem:induction} that $T(n) \leq  b\cdot C\cdot n^{1+a}$, where $a$ is a constant, and $b$ and $C$ are super-constant in $n$, as long as a few conditions hold. Corollary~\ref{lem:instance} lists values that satisfy these conditions, concluding that for all $0<\varepsilon<1/2$, there are valid assignments to the variables that achieve $T(n)\leq  n^{3/2+\varepsilon}$.\footnote{It can be shown that this is in fact tight.} Furthermore, Corollary~\ref{lem:instance} ensures that $C$ is super-constant (and in particular polynomial in $n$), implying that with probability at least $\frac{{n\choose 2}-n^{3/2+\varepsilon}}{{n\choose 2}}$, the handle of the randomly planted broom will be high-degree for super-constant $C$. It can therefore only be selected with a sub-constant probability.

\begin{lemma}\label{lem:induction}
Consider the following recurrence when $C \leq (n-1)/2$.
\begin{align*}
    T(n) \leq  \max_{\vec{x}\in \mathbb{N}_{\geq 0}^{n-1}, \sum_i x_i = n, } \min\bigg\{ &\sum_i T(x_i) + 2C(n-1)~~,~~ T(n-1) + 2(n-1-\max_i\{x_i\}) \bigg\}.
\end{align*}
with a base case of $T(n)=n(n-1)/2$ when $(n-1)/2<C$. For all $N$, $T(N) \leq  b\cdot C\cdot N^{1+a}$, as long as
\begin{enumerate}\itemsep0em%[nosep]
    \item $a\in(0,1)$ is a constant;
    \item $C$ is a super-constant function of $N$;
    \item $b$ is a super-constant function of $N$ such that $b(N)\geq 1$ for all $N$;
    \item  for all $n<N$, the following is satisfied: $\frac{2(n-1)}{ab n^{1+a}}<\frac{(1+a)bC}{2n^{1-a}}.$
\end{enumerate}
\end{lemma}

As an immediate corollary, we get the following.

\begin{corollary}\label{lem:instance}
Let $T(n)$ be defined as in Lemma~\ref{lem:induction}. Then for all $0<\varepsilon<1/2$, $T(n)\leq  n^{3/2+\varepsilon}$.
\end{corollary}
\begin{proof}
Fix an $\varepsilon$, and let $a=\varepsilon/3$, $C=n^{\varepsilon/3}$, and $b=\frac{4}{a}\cdot\left(\frac{n}{C}\right)^{1/2+\varepsilon/3}$.
\end{proof}

Now we can complete the proof of Theorem~\ref{thm:partition-neg}. Corollary~\ref{lem:instance} together with Proposition~\ref{prop:recurrence} establishes that for any $\varepsilon > 0$, there are at most $n^{3/2+\varepsilon}$ edges with degree at most $C:=n^{\varepsilon/3}$. This means that with probability $1-n^{-1/2+\varepsilon}$, a randomly selected edge $(u,v)$ of the complete graph has degree at least $n^{\varepsilon/3}$. 
Conditioned on $(u,v)$ having high-degree,  we know that $n^{\varepsilon/3}$ edges of the max-weight spanning tree are in the same partition as $(u,v)$. Therefore, at least one of them is selected with probability at most $n^{-\varepsilon/3}$. Setting $\varepsilon=3/8$, we conclude that except with probability $n^{-1/8}$, there is some edge selected with probability at most $n^{-1/8}$, and therefore no randomized partition algorithm can be $\omega(n^{-1/8})$-probability-competitive.

\section{Conclusion}
We study the limitations of simple frameworks for the matroid secretary problem. We show that a class of natural greedy algorithms cannot be constant-utility-competitive, and the class of randomized partition algorithms cannot be constant-probability competitive. This helps explain why such algorithms (especially greedy ones) have faced barriers to resolving the matroid secretary conjecture, and also helps narrow future work towards frameworks with potential.

This agenda leaves open much future work. For example, our greedy framework did not incorporate the virtual algorithm of~\cite{babaioff2007knapsack}. Is there a natural generalized greedy framework which does? Do similar lower bounds hold against that framework? Similarly, replacing $\mathcal{I}$ with a disjoint union of $1$-uniform matroids is a very special case of a framework which replaces $\mathcal{I}$ with a disjoint union of arbitrary matroids (and potentially runs a more complex greedy algorithm than Dynkin's)~\cite{lachish2014log, FSZ-SODA15,  huynh2016matroid,feldman2016online} . Do similar lower bounds hold against this framework? There are also other frameworks not considered in this paper (such as the ``forbidden sets'' approach of~\cite{SotoTV-SODA18}) which should be investigated as well.

{\small
\bibliographystyle{alpha}
\bibliography{main}
}
\newpage
\appendix
% !TeX root = main.tex
% !TEX root = main.tex

% the comments that used to be in this file can be found in WINE 2021 folder. 
\section{Missing Proofs} \label{sec:missingProofs}

\subsection{Proofs from Section~\ref{sec:hat}}\label{app:hat}

\begin{proof}[Proof of Observation~\ref{obs:aa-win}]
The algorithm wins iff it accepts the infinity edge. Each $(AA)$ claw forms a cycle with the infinity edge, so it is infeasible to accept the infinity edge after having formed an $(AA)$ claw. By property~(\ref{prop3}), once an $(AA)$ is formed, it will never get removed from $I_t$. Therefore, if there is an $(AA)$ claw before $t_\infty$, the algorithm definitely loses.

In the other direction, if there is no $(AA)$ claw right before $t_\infty$, adding the infinity edge to $I_{t_\infty}$ creates at most one cycle, as otherwise $I_{t_\infty}$ would not be independent. This cycle consists of the infinity edge and a claw with at least one $S$ edge. The algorithm can win by discarding this $S$ edge from $I_{t_\infty}$ and accepting the infinity edge.
\end{proof}

\begin{proof}[Proof of Lemma~\ref{obs:blocker}]

Suppose the upper edge $e_i$ of a $(-\,A)$ claw arrives at time $t$.  By property (iii), $e_i$ is accepted iff it is in the max-weight basis of $I_t\cup \{e_i\}$ contracted by $A_t$. If $I_t$ contains at most one edge from every other claw, it would be feasible to add $e_i$ so it is in the MWB and will be accepted. Suppose there exists some claw $j$ with both edges in $I_t$. Note that there can be at most one such claw by independence of $I_t$ (property~(\ref{prop3})), so there is a unique claw $j$ with both edges in $I_t$.

We consider all four cases for claw $j$. First, note that $j$ cannot be an $(A\,A)$ claw by assumption. Assume $j$ is a $(S\,S)$ claw. Claws $i$ and $j$ form a cycle in $I_t\cup \{e_i\}$, and the MWB of $I_t\cup \{e_i\}$ contracted by $A_t$ contains the heaviest two of $e_j$, $e'_j$ and $e_i$, namely $e_i$ and $e_j$. So $e_i$ will be accepted.

Next, assume $j$ is a $(A\,S)$ claw. Claws $i$ and $j$ form a cycle in $I_t\cup \{e_i\}$, and the MWB of $I_t\cup \{e_i\}$ contracted by $A_t$ contains the heavier of $e'_j$ and $e_i$, namely $e_i$. So $e_i$ will be accepted.

Next, assume $j$ is a $(S\,A)$ claw. Claws $i$ and $j$ form a cycle in $I_t\cup \{e_i\}$, and the MWB of $I_t\cup \{e_i\}$ contracted by $A_t$ contains the heavier of $e_j$ and $e_i$, namely $e_i$ iff $i<j$. In other words, we have shown that the upper edge of a $(-\,A)$ claw $i$ is not accepted at time $t$ iff $I_t$ contains a $(S\,A)$ claw to the left of $i$.
\end{proof}

\begin{proof}[Proof of Lemma~\ref{obs:unprotected}]
When the lower edge $e'$ of a $(S\,-)$ claw arrives, its upper edge has to be in $I_t$ for $I_t$ to be spanning (property~(\ref{prop3})). Adding $e'$ to $I_t$ thus forms a cycle in $I_t$ with the previous $(SA)$. The max-weight basis of $I_t\cup \{e'\}$ contracted by $A_t$ will contain the two heaviest edges of this length-4 cycle other than the already accepted edge, namely the two sample edges. Since the MWB does not contain $e'$, the algorithm cannot accept it by property~(\ref{prop2}).
\end{proof}

\subsection{Proofs from Section~\ref{sec:supergreedy}}\label{app:supergreedy}

\begin{proof}[Proof of Theorem~\ref{thm:supergreedy}]
Let $\alpha>1$ be a constant. We will consider a hat with $n+2$ vertices, where the weight function satisfies $\frac{1}{2\alpha}<w(e_i)<\frac{1}{\alpha}$ and $\frac{1}{3\alpha}<w(e'_i)<\frac{1}{2\alpha}$ for all $i\in[n]$ in addition to $e_1 >\ldots> e_n > e_1' >\ldots> e'_n$.\footnote{The weight assignments in the proof provided by Babaioff \textit{et al.} are slightly different. The proof does not need the property that the heavier upper edges are paired with the heavier lower edges, so the proof would go through without the assumption that $e_1 >\ldots> e_n > e_1' >\ldots> e'_n$. However, we will use this additional assumption for the convenience of talking about the location of claws as a proxy for their relative weights.} Furthermore, we let the infinity edge have weight $n+1$. If the algorithm fails to select the infinity edge, it can accept a total weight smaller than $(m+1)/\alpha$, and thus its competitive ratio is less than $\alpha$. We will again use the algorithm ``losing'' interchangeably with its failing to accept the infinity edge.

If the infinity edge arrives in the sample stage, the algorithm loses immediately. Suppose the infinity edge arrives after $T$ for the rest of the proof. The proof sketch is as follows. We would like to show that an $(AA)$ claw will be accepted with high probability before the infinity edge arrives. To do so, we show with high probability, no blocker is formed before the first time the upper edge of a $(-\,A)$ claw arrives. This is due to the fact that $(SS)$ claws in $I$ prevent blockers from forming to their right, and after the sampling stage, $I_T$ will contain an $(SS)$ claw that is very far to the left with high probability. Furthermore, an $(SS)$ claw that is far to the left is hard to discard and will last for a long enough time that an $(AA)$ claw will form with high probability.

More formally, we will consider the interval $L=[T,\min(T+n^{-1/3},1)]$, and show that an $(AA)$ claw will already form during $L$ with high probability, whereas the infinity edge does not arrive during that interval with high probability, \textit{i.e.} $t_\infty\not\in L$.

Recall from Observation~\ref{obs:aa-win} that the algorithm wins iff no $(AA)$ claws form before $t_\infty$. By Lemma~\ref{obs:blocker}, the only way to protect against accepting certain $(AA)$ claws (\textit{i.e.} ones in which the upper edge arrives after the lower edge) is by forming a blocker to their left. We will show that with high probability, the algorithm does not build a blocker during $L$. To see why, note that immediately after the sampling stage, $I_T$ contains the max-weight-basis of the sample elements. $I_T$ will contain at most one claw by independence, namely the left-most $(SS)$ claw, say claw $i$. Each claw has a $1/4$ chance of being type $(SS)$ at $T$, independently of every other claw. Therefore, $i$ is geometrically distributed with success probability $1/4$, and thus $i=O(\log n)$ with probability at least $(1-1/n^2)$.

Since the algorithm can only kick out the lightest edge in a cycle, the only way for this $(SS)$ claw to get discarded is for the algorithm to accept an edge that forms a claw to its left. The reason is that among the edges to the right of $i$, only upper edges are heavy enough to kick out an edge of claw $i$ when accepted. Consider any such upper edge in location $j>i$. For it to be possible to discard the $(SS)$ claw $i$, claw $j$ must form a cycle with claw $i$. Therefore, the lower edge in location $j$ must be either an $A$ or an $S$. In the former case, the algorithm has already lost, and in the latter case, the algorithm kicks out the the lower edge in location $j$ over either edge of the $(SS)$ claw in location $i$.

However, the probability that an edge to the left of $i$ arrives during $L$ is no more than $O(n^{-1/3}\log n)$ by union bound. Therefore, with high probability, the $(SS)$ claw $i$ will survive for the entire duration of $L$. Conditioned on this event, no blocker can be form in $I_t$ for $t\in L$: The only potential blockers must be to right of $i$ (since we are conditioning on no edge to the left of $i$ arriving in $L$), and the lower edge of any such potential blocker claw $j$ is lighter than all the sample elements in $I_t$ with which it forms a cycle (\textit{i.e.} both edges of claw $i$ and the upper edge of claw $j$), and thus will not be accepted.

Let $\mathcal{E}_i$ refer to the event that $t(e'_i),t(e_i)\in L$ and $t(e'_i) < t(e_i)$, and note that conditioned on no blocker forming during $L$, any $\mathcal{E}_i$ causes the algorithm to fail. The expected number of $\mathcal{E}_i$s that happen is at least $n|L|^2=n^{1/3}/2$. Additionally, since the $\mathcal{E}_i$s only depend on the draws $t(e_i)$ and $t(e'_i)$ and are independent of each other, the number of $\mathcal{E}_i$s that happen is concentrated around its expectation by a standard Chernoff bound. Therefore, with high probability, some $\mathcal{E}_i$ happens and the algorithm loses. 
\end{proof}

\subsection{Proofs from Section~\ref{sec:all-greedy}}\label{app:all-greedy}

\begin{proof}[Proof of Lemma~\ref{lem:AA}]

From the left-most $x$ claws, in expectation $x\ell^2$ of them satisfy $T<t(e'_i)<T+\ell$ and $T<t(e_i)<T+\ell$. With exponentially small probability fewer than $x\ell^2/2$ such claws exist.  We count that small probability towards the algorithm's success, and assume $x\ell^2/2$ of those claws exist. Each of them satisfies $T<t(e'_i)<t(e_i)<T+\ell$ with probability 1/2 independently of other claws. Therefore the probability of existence of a claw satisfying $T<t(e'_i)<t(e_i)<T+\ell$ is at least $1-(1/2)^{\ell^2x/2}$.
\end{proof}

\begin{proof}[Proof of Lemma~\ref{lem:bullettwo}]
 
We first upperbound the number of claws that are candidates for becoming a blocker. For a blocker to the left of $x$ to be accepted, there must be a $(S\,-)$ claw to the left of $x$ at time $T$, with the lower edge arriving during $[T,T+\ell]$. Conditioning on Lemma~\ref{lem:AA} (i.e. existence of a $(-\,-)$ claw to the left of $x$ at time $T$) only decreases the number of possible $(S\,-)$ claws, so we can remove the conditioning. For a blocker to be accepted, its lower edge must arrive after $T$ and during an unprotected period (Lemma~\ref{obs:unprotected}). From the $x$ candidate lower edges, in expectation $x\ell$ of them arrive in $[T, T+\ell]$. With exponentially small probability more than $2x\ell$ edges arrive during that time. We count that small probability towards the algorithm's success. We will call each such lower edge a \emph{nice} edge.

Suppose the algorithm chooses to go unprotected during the window $W=[T+\ell-R, T+\ell-R+dy]$ where $R$ represents the remaining time till $T+\ell$ and $dy$ is the length of the window. Suppose $X'$ nice edges (out of at most $2x\ell$ nice edges) are yet to arrive at $T+\ell-R$. Any nice edge arriving after $T$ lands in $W$ with probability $dy/R$, so $W$ fails to include any nice edges with probability at least $(1-dy/R)^{X'}$.

Lemma~\ref{lem:ratio-concentration} shows $X'(R)/R\leq 4x$ with probability at least $1-2x\ell e^{-\frac{2x\ell}{3}}$. 
\begin{lemma}\label{lem:ratio-concentration}
Let $X'(R)$ be the number of nice edges yet to arrive at time $T+\ell-R$. Then $\sup_{R\in[T,T+\ell]}\{\frac{X'(R)}{R}\}\leq 4x$ with probability at least $1-2x\ell e^{-\frac{2x\ell}{3}}$.
\end{lemma}
\begin{proof}
Note that $X'$ can only take values in $\{0,1,\ldots,2x\ell\}$\footnote{If there are fewer than $2x\ell$ nice edges, our result can only get stronger.}, and let $c_1,\ldots,c_{2x\ell}$ be the set of nice edges ordered arbitrarily (and not as a function of their arrival time).
Let $T+\ell-R_i$ denote the arrival time of $c_i$. The set of local maxima of $X'(R)/R$ is exactly $\{X'(R_i)/R_i : i\in[2x\ell]\}$, and it suffices to show $\max_{i} X'(R_i)/R_i \leq 4x$ with high probability.

The arrival times of nice edges are distributed uniformly in $[T,T+\ell]$. For every $i,j\in[2x\ell], i\not=j$, the edge $c_j$ arrives in $[T+\ell-R_i, T+\ell]$ with probability $R_i/\ell$ and in $[T,T+\ell-R_i]$ otherwise. So for every $i\in[2x\ell]$, $X'(R_i)$ is distributed according to Binom($2x\ell,R_i/\ell$). By Chernoff bounds, 
\begin{align*}
\mathbb{P}\left[\frac{X'(R_i)}{R_i}\geq 4x\right] \leq e^{-\frac{2x}{3}}.
\end{align*}

Taking a union bound over all $i\in[2x\ell]$ gives the desired result.
\end{proof}
By Lemma~\ref{lem:ratio-concentration}, $(1-\frac{dy}{R})^{X'}$ with high probability is at least $(1-\frac{dy}{R})^{R(4x)}$, which is at least $(1-dy)^{4x}$ for $dy \leq R$. Finally, recall that the algorithm is free to pick a set of total measure $y$ of unprotected time distributed however it pleases. Since $(1-dy)^{4x}$ is concave in $dy$, the total failure probability incurred by the algorithm is at least $(1-y)^{4x}$.
Combining this with Lemma~\ref{lem:ratio-concentration}, algorithm's overall failure probability is at least 
\begin{align*}
(1-2x\ell e^{-\frac{2x}{3}})(1-y)^{4x}.
\end{align*}
\end{proof}

\begin{proof}[Proof of Lemma~\ref{lem:bulletone}]

The proof approach is as follows. The algorithm \emph{has} to allocate \emph{at least} $y$ total unprotected time in the window $[T,T+\ell]$. Any unprotected time carries the risk of accepting the upper edge of a $(-\,A)$ claw \emph{anywhere}. Given the history of arrivals, the algorithm can make adaptive decisions about when to go unprotected, with the hope of being unprotected when upper edges of $(-\,A)$ claws are least likely to arrive. We will make the algorithm's job easier by allowing it to go unprotected near the beginning and end of the window $[T,T+\ell]$ ``for free'', that is we ignore the failure cases where the upper edge of a $(-\,A)$ claw arrives during those windows. This modification only undercounts the algorithm's failure probability, making our lower bound stronger.

In particular, we again focus on claws satisfying  $T < t(e'_i) < t(e_i) <T+\ell$, noting that any such claw that arrives in an unprotected period would result in the algorithm failing.\footnote{While this is not an exhaustive list of ways the algorithm could fail, it suffices to show that this particular failure case is already likely.} Each claw $i\in[n]$ satisfies $T < t(e'_i) < t(e_i) <T+\ell$ independently with probability $\ell^2/2$. The expected number of such claws is therefore $\ell^2n/2$. Standard Chernoff bounds imply that except for an exponentially small probability, the number of such claws does not drop below $\ell^2n/4$. Therefore, even if the algorithm manages to accept the infinity edge for \emph{all} the arrival orders in which the number of such claws drops below $\ell^2n/4$, its competitive ratio increases by only an exponentially small amount. We count that small probability towards the algorithm's success and assume there are at least $\ell^2n/4$ such claws. Uniformly at random select $\ell^2n/4$ of such claws and call them \emph{potential AA}. We say that a potential AA $(e_i,e'_i)$ \emph{arrives at} time $t=t(e_i)$.

Suppose the algorithm chooses to go unprotected during the window $W=[T+\ell-R, T+\ell-R+dy]$. 
We say a potential AA is \emph{pending at time $t$} if its lower edge arrives in $[T,t]$ and its upper edge arrives in $[t,T+\ell]$. Let $A'(R)$ be the number of pending claws at time $T+\ell-R$. We claim that the probability of the algorithm failing is lower bounded by $1-(1-dy/R)^{A'(R)}$. This is because any potential AA pending at time $T+\ell-R$ lands in $W$ with probability $dy/R$. We show that for sufficiently large $y$, $A'(R)$ is concentrated around its expectation. The expected value of $A'(R)$ is $n(\ell-R)R/\ell^2$.

By a Chernoff bound,
$
\mathbb{P}\left[A'(R) \leq \frac{n(\ell-R)R}{2\ell^2}\right] \leq e^{-\frac{n(\ell-R)R}{8\ell^2}}.
$
Similar to Lemma~\ref{lem:ratio-concentration} the local minima for $\frac{A'(R)}{R}$ occur when a potential AA arrives. 
Taking a union bound over all $\ell^2n/4$ potential AA's, we conclude that $A'(R)$ never exceeds $ \frac{n(\ell-R)R}{2\ell^2}$ with probability at least $1-\frac{\ell^2 n}{4} e^{-\frac{n(\ell-R)R}{8\ell^2}}$.

We will make the algorithm's job easier by allowing it to go unprotected during both $[T,T+\varepsilon]$ and $[T+\ell-\varepsilon, T+\ell]$ without accepting any $(AA)$ claws then. The algorithm still has to allocate $y-2\varepsilon$ unprotected time within $[T+\varepsilon, T+\ell-\varepsilon]$. Any such window $W=[T+\ell-R, T+\ell-R+dy]$ causes the algorithm to fail with probability at least $1-(1-dy/R)^{A'(R)}$, which with probability at least $1-e^{-\frac{n(\ell-R)R}{8\ell^2}}$ is at least
\[1-(1-dy/R)^{\frac{n(\ell-R)R}{2\ell^2}} \geq 1-\left(1-\frac{dy}{\ell}\right)^{\frac{n\varepsilon(\ell-\varepsilon)}{2\ell^2}}, \]
where the inequality follows from $0 \leq \varepsilon \leq R \leq \ell-\varepsilon\leq \ell$.
For $\ell=n^{-0.1}$, the right-hand-side expression is concave in $dy$, so the algorithm is best off going unprotected in a consecutive interval of length $y-2\varepsilon$, incurring an overall failure probability of
\[ \left(1-(\frac{n \ell^2}{4}) e^{-\frac{n(\ell-\varepsilon)\varepsilon}{8\ell^2}}\right)\left(1-\left(1-\frac{y-2\varepsilon}{\ell}\right)^{\frac{n\varepsilon(\ell-\varepsilon)}{2\ell^2}}\right).\]
Setting $\varepsilon=n^{-0.4}/4$, the left term vanishes for large $n$, so the failure probability of the algorithm is lower bounded by 
\[
1-\left(1-\frac{2y-n^{-0.4}}{2\ell}\right)^{\frac{n^{0.6}(4\ell-n^{-0.4})}{32\ell^2}}.
\]
\end{proof}

\begin{proof}[Proof of Lemma~\ref{lem:combine}]

Note that $f(y)$ is decreasing and $g(y)$ is increasing in $y$. 
Suppose that the optimum value of $y$ for the algorithm is $y_*$; i.e., $y_* = \argmin_{y } \max \{f(y),g(y)\}$.
For any $y_0$, we have:
\begin{align*}
\min\{g(y_0),f(y_0)\} &\leq \max\big\{ \min\{g(y_*),f(y_0)\},\,\min\{g(y_0),f(y_*)\}\big\}\\
&\leq \max\{g(y_*), f(y_*)\} ~=  \min_{y}\big\{\max\{ g(y), f(y)\}\big\}.
\end{align*}
The first inequality holds because $g$ is increasing and $f$ is decreasing, so in the case that $y_0<y_*$, the first term in the max is larger than the left-hand-side, and when $y_0>y_*$, the second term in the max is larger than the left-hand-side. The second inequality holds since dropping mins can only make the value larger, and the last equality holds by definition of $y_*$. We will pick a value of $y_0$ that gives a sufficiently good lower bound to $\min_{y}\{\max\{ g(y), f(y)\}\}$.

In particular, we set $x=n^{0.3}$, $y_0=n^{-0.4}$, and $\ell=n^{-0.1}$. Then, as $n\rightarrow\infty$, we have
\begin{align*}
	f(y_0) &=(1-2x\ell e^{-\frac{2x}{3}})(1-y_0)^{4x}\left(1-(\frac{1}{2})^{\ell^2x/2}\right)\\
	g(y_0) &=
	 1-\left(1-\frac{2y_0-n^{-0.4}}{2\ell}\right)^{\frac{n^{0.6}(4\ell-n^{-0.4})}{32\ell^2}}.
\end{align*}
By substituting $x$, $y_0$, and $\ell$, for $n \rightarrow \infty$, $(1-2x\ell e^{-\frac{2x}{3}})$ and $(1-(1/2)^{\ell^2x/2})$ in $f(y_0)$ tend to $1$. Therefore, we only need to reason about $(1-y_0)^{4x}$ and $(1-\frac{2y_0-n^{-0.4}}{2\ell})^{\frac{n^{-0.6} \ell}{32}}$. Note that $(1-y_0)^{4x}$ converges to $e^{-4xy_0}$ which is $e^{4n^{-0.1}}$ and tends to $1$. Similarly $1-(1-\frac{2y_0-n^{-0.4}}{2\ell})^{\frac{n^{0.6}(4\ell-n^{-0.4})}{32\ell^2}}$ converges to $1-e^{-\frac{(2y_0-n^{-0.4})n^{0.6}(4\ell-n^{-0.4})}{64\ell^3}}$ which goes to 1 in the limit as $n$ goes to infinity.

Since both lower bounds converge to $1$, for any constant $0 <\alpha \leq 1$, there exists $n_0$ such that for $n>n_0$, the algorithm will fail with probability $1-\alpha$ and cannot be constant utility-competitive, thus completing the proof.

\end{proof}

\subsection{Proofs from Section~\ref{ch:partition}} \label{sec:ProofKorulaPal}

\begin{proof}[Proof of Theorem~\ref{thm:KorulaPal}] For completeness, we show the feasibility and  competitiveness of the Korula-Pal algorithm.

\textbf{Feasibility.}
 Assume, for contradiction, that the algorithm accepts a cycle on nodes $v_1,\ldots, v_k$. Consider the node $v$ that appears earliest in the offline ordering $\sigma$, {i.e.}, $v = \arg\min \sigma(v_i)$. Then, two edges are accepted within $E_v$, contradicting the feasibility of the output of Dynkin's algorithm.

\textbf{Competitiveness.}
Let $\ALG_v$ denote the weight of the edge accepted in the instance of Dynkin's algorithm running on $E_v$. By the competitiveness of Dynkin's algorithm, we know that $\E[\ALG_v]\geq\E[\max_{i\in E_v} w(i)/\e]$. Let $\ALG$ be the total weight accepted by the algorithm. Since $\{E_v\}$ forms a partition, we know that
$$\E[\ALG] =\E\left[\sum_v \ALG_v\right]\geq\E\left[\sum_v \max_{i\in E_v} \frac{\{w(i)\}}{\e}\right]=\sum_v\E\left[ \max_{i\in E_v} \frac{\{w(i)\}}{\e}\right],$$
so it suffices to bound $\E\left[\max_{i\in E_v} \{w(i)\}\right]$.
\\
Fix a max-weight basis of $G$ (a spanning forest of maximum total weight $\OPT$), called $T$. Let $v$ be a leaf in $T$, and let $i$ be the edge incident to it. Then, with probability $1/2$, $i\in E_v$. So $\E[\max_{i\in E_v} w(i)]\geq \frac{1}{2}\cdot w_{i}/2$ for all $v$ that are leaves of $T$. We can iteratively remove the leaves of $T$, and repeat the above reasoning to obtain a similar bound for all nodes $v$ in $T$. Putting everything together, we have
$$\E[\ALG] \geq \sum_v\E\left[ \max_{i\in E_v} \frac{\{w(i)\}}{\e}\right]\geq \sum_{i\in T} \frac{w(i)}{2\e} = \frac{\OPT}{2\e}.$$
\end{proof}

In the above proof, importantly, note that it does \emph{not} claim that any particular edge of the max-weight spanning tree is selected with high probability. Each edge just serves as a witness that the expected max-weight edge adjacent to each particular node is high (but in fact, it could be much higher, and this is why the algorithm is not probability-competitive).

\subsection{Proof of Lemma~\ref{lem:induction}} \label{sec:PfLemmaRandPart}

\begin{proof}[Proof of Lemma~\ref{lem:induction}]
We will proceed by induction on $n$.

\textbf{Base Case.} When $(n-1)/2<C$, we have $T(n)=n(n-1)/2<C\cdot n<C\cdot n^{1+a}\leq  b\cdot C\cdot n^{1+a}$, where in the last equality we have used the assumption that $b>1$.

\textbf{Inductive Step.} For simplicity of notation, permute the partitions so that $x_1 = \arg\max_i\{x_i\}$. We consider two possibilities depending on how large $x_1$ is compared to $(1-\gamma)n$ for some $0<\gamma<1$ that we will choose later. 
\begin{enumerate}%[nosep]
    \item $x_1< (1-\gamma)n$: In this case, we have
    \begin{align*}
        T(n) &\leq  \sum_i T(x_i) + 2C(n-1)\\
        &\leq  \sum_i bC(x_i)^{1+a} + 2C(n-1)\\
        &\leq  bC\big(((1-\gamma)n)^{1+a} +(\gamma n)^{1+a}\big) + 2C(n-1)\\
        &= bCn^{1+a}\left((1-\gamma)^{1+a} +\gamma^{1+a}\right) + 2C(n-1)\\
        &= bCn^{1+a}\left(1-(a+1)\gamma+O(\gamma^2) +\gamma^{1+a}\right) + 2C(n-1) -4+2\gamma n\\
        &\leq  bCn^{1+a} - \gamma abCn^{1+a} + O(\gamma^2) + 2C(n-1)
    \end{align*}
    where the first inequality uses the first term in the max in our recursive bound on $T(n)$; the second inequality uses the inductive hypothesis and the assumption that $x_1$ is large; the third inequality uses convexity of $f(x)=x^{1+a}$ for $a>1$ as shown in Lemma~\ref{lem:convex}; the equality on the fourth line is simple algebraic manipulation; the equality on the fifth line follows from the Taylor Expansion of $(1-\gamma)^{1+a}$; and the last equality follows from algebraic manipulation as well as the fact that $a>0,\gamma<1,\Rightarrow \gamma>\gamma^{1+a}$.\\
    \noindent
    It therefore suffices to ensure that
    \begin{align*}
        0\geq - \gamma abCn^{1+a} + O(\gamma^2) + 2C(n-1),
    \end{align*}
    Since the error term $O(\cdot)$ is non-negative, it is okay to drop it and ensure that $\gamma\geq \frac{2(n-1)}{ab n^{1+a}}$. We instead enforce the stronger result \fbox{$\gamma\geq \frac{2}{ab n^a}$}.
    
    \item $x_1\geq (1-\gamma)n$: In this case, we have
    \begin{align*}
        T(n) &\leq  T(n-1)+2((n-1)-1-x_1)\\
        &\leq  bC(n-1)^{1+a}+2(n-2-x_1-n(1-\gamma))\\
        &= bCn^{1+a}\left(1-\frac{1}{n}\right)^{1+a}-4+2\gamma n\\
        &= bCn^{1+a}\left(1-\frac{1+a}{n}+O\left(\frac{1}{n^2}\right)\right) -4+2\gamma n\\
        &= bCn^{1+a}-(1+a)bCn^{a}+O\left(\frac{bC}{n^{1-a}}\right) -4+2\gamma n,
    \end{align*}
    where the first inequality uses the second term in the max in our recursive bound on $T(n)$; the second inequality uses the inductive hypothesis and the assumption that $x_1$ is large; the equality on the third line is simple algebraic manipulation; the equality on the fourth line follows from the Taylor Expansion of $(1-1/n)^{1+a}$; and the last equality is simple algebraic manipulation.\\
    \noindent
    To complete the inductive step, we need to make sure the RHS is at most $bCn^{1+a}$, which holds as long as
    \begin{align*}
        0\geq -bC(1+a)n^{a}+O\left(\frac{bC}{n^{1-a}}\right) -4+2\gamma n,
    \end{align*}
    which is the same as
    \begin{align*}
        \gamma&\leq  \frac{1}{2n}(4+(1+a)bCn^a)+O\left(\frac{bC}{n^{2-a}}\right)\\
        &= \frac{2}{n}+\frac{(1+a)bC}{2n^{1-a}}+O\left(\frac{bC}{n^{2-a}}\right).
    \end{align*}
    Since both the error term $O(\cdot)$ and $2/n$ term are non-negative, it is enough to ensure \fbox{$\gamma\leq \frac{(1+a)bC}{2n^{1-a}}$}.
\end{enumerate}
The two cases above give a lower-bound and an upper-bound on $\gamma$, and the inductive step holds for any assignment of $a$, $b$, $C$ and $\gamma$ that satisfies $$\frac{2(n-1)}{ab n^{1+a}}\leq \gamma\leq \frac{(1+a)bC}{2n^{1-a}}.$$
Indeed, by the fourth assumption of this Lemma, the lower-bound is strictly smaller than the upper bound for all $n$, and thus there exists some $\gamma$ (perhaps one depending on $n$) that satisfies both inequalities, thus completing the proof.
\end{proof}

\begin{lemma}\label{lem:convex}
Let $\gamma\in(0,1/2)$ and $a\in (0,1)$. Consider the following optimization problem
\begin{equation*}
\begin{aligned}
& \underset{\vec{x}\in \mathbb{N}_{\geq 0}^n}{\text{minimize}}
& & \sum_{i} x_i^{1+a} \\
& \text{subject to}
& & \sum_{i} x_i = n,\\
& & & 0 \leq  x_i \leq  (1-\gamma) n, \; i = 1, \ldots, m.
\end{aligned}
\end{equation*}
The optimum value of this program is $\left((1-\gamma)n\right)^{1+a} + \left(\gamma n\right)^{1+a}$.
\end{lemma}
\begin{proof}
Note that $\vec{y}=\big((1-\gamma)n, \gamma n, 0, \ldots, 0\big)$ is a feasible solution achieving the upper bound. Suppose another feasible solution $\vec{x}$ achieves the optimum value that is higher than $\left((1-\gamma)n\right)^{1+a} + \left(\gamma n\right)^{1+a}$. Since the order of the coordinates don't affect the objective or the constraints, assume $x_1\geq \ldots\geq x_n$ without loss of generality.

We claim that $x_2\not=0$. Otherwise, the only non-zero coordinate of $\vec{x}$ is $x_1$, so $\sum x_i = x_1\leq  (1-\gamma)n < n$, violating the first constraint.

Suppose $x_1\not=y_1$. Then $x_1<y_1$ by the second constraint. It is therefore possible to weakly improve the value of $\vec{x}$ by changing it into $\vec{x}'=(x_1 +\delta, x_2-\delta, x_3,\ldots,x_n)$ for some small enough $\delta$. This is because $f(x)=x^{1+a}$ is strictly convex for $a>0$, so its first derivative is non-decreasing, and by the fundamental theorem of calculus, we have
\begin{align*}
    f(x_1 +\delta) - f(x_1) = \int_{x_1}^{x_1+\delta} f'(x) dx \geq \int_{x_2-\delta}^{x_2} f'(x) dx \geq f(x_2) - f(x_2-\delta).
\end{align*}
Rearranging the terms gives $f(x_1 +\delta) + f(x_2-\delta) \geq  f(x_1) + f(x_2)$ as desired. Therefore, we can assume $x_1=y_1$ for any candidate solution $\vec{x}$ that achieves a higher value than $\vec{y}$. The problem therefore reduces to a new optimization problem with constraints $\sum x_i = \gamma n$ and $0\leq  x_i \leq  \gamma n$. By the same logic as above, we can assume without loss that any optimal solution puts all the weight on the first coordinate and 0 on the rest of the coordinates (making the second constraint tight on every coordinate). Therefore, we have $x_2=y_2$. Since $x_1+x_2=n=y_1+y_2$, we must have $x_j=y_j=0$ for all $j\not\in\{0,1\}$, implying $\vec{x}$ and $\vec{y}$ achieve the same value, which is a contradiction.
\end{proof}

\end{document}